\documentclass[smallextended]{svjour3}

\usepackage{cite}
\usepackage{url}
\usepackage{amsmath}
\usepackage{amssymb}
\usepackage{graphicx}
\usepackage[tight]{subfigure}
\usepackage[english]{babel}
\usepackage{booktabs}
\usepackage[numbers]{natbib}
\usepackage{xspace}
\usepackage[ruled, vlined, nofillcomment]{algorithm2e}

\usepackage{tikz}
\usetikzlibrary{shapes,positioning,fit,calc}
\usepackage{pgfplots}

\usepgfplotslibrary{external}
\newcommand{\set}[1]{\left\{#1\right\}}
\newcommand{\pr}[1]{\left(#1\right)}
\newcommand{\fpr}[1]{\mathopen{}\left(#1\right)}
\newcommand{\spr}[1]{\left[#1\right]}
\newcommand{\fspr}[1]{\mathopen{}\left[#1\right]}

\newcommand{\abs}[1]{{\left|#1\right|}}

\newcommand{\np}{\textbf{NP}}

\newcommand{\funcdef}[3]{{#1}:{#2} \to {#3}}

\newcommand{\reals}{{\mathbb{R}}}

\DeclareRobustCommand{\dispfunc}[2]{%
  \ensuremath{%
  \ifthenelse{\equal{#2}{}}%
    {{#1}}%
    {{#1}\fpr{#2}}}}

\newcommand{\comment}[1]{}

\newcommand{\lab}[1]{\dispfunc{\mathit{lab}}{#1}}
\newcommand{\supp}[1]{\dispfunc{\mathit{sup}}{#1}}
\newcommand{\pre}[1]{\dispfunc{\mathit{pre}}{#1}}
\newcommand{\mach}[1]{\dispfunc{\mathit{M}}{#1}}
\newcommand{\greedy}[1]{\dispfunc{\mathit{gr}}{#1}}
\newcommand{\pind}[1]{\dispfunc{\mathit{p_{\mathit{ind}}}}{#1}}
\newcommand{\pprt}[1]{\dispfunc{\mathit{p_{\mathit{prt}}}}{#1}}
\newcommand{\blockp}[1]{\dispfunc{C_p}{#1}}
\newcommand{\blocks}[1]{\dispfunc{C_s}{#1}}
\newcommand{\rind}[1]{\dispfunc{r_{\mathit{ind}}}{#1}}
\newcommand{\rpart}[1]{\dispfunc{r_{\mathit{prt}}}{#1}}
\newcommand{\rank}[1]{\dispfunc{r}{#1}}
\newcommand{\diag}[1]{\dispfunc{diag}{#1}}

\newcommand{\mean}[2]{\operatorname{E}_{#1}\fspr{#2}}

\let\oldexample\example
\renewcommand\example{\oldexample\normalfont}

\SetKwComment{tcpas}{\{}{\}}
\SetCommentSty{textnormal}
\SetArgSty{textnormal}
\SetKw{False}{false}
\SetKw{True}{true}
\SetKw{Null}{null}
\SetKwInOut{Output}{output}
\SetKwInOut{Input}{input}
\SetKw{AND}{and}
\SetKw{OR}{or}
\SetKw{Break}{break}

\pgfdeclarelayer{background}
\pgfdeclarelayer{foreground}
\pgfsetlayers{background,main,foreground}

\definecolor{yafaxiscolor}{rgb}{0.3, 0.3, 0.3}

\definecolor{yafcolor1}{rgb}{0.4, 0.165, 0.553}
\definecolor{yafcolor2}{rgb}{0.949, 0.482, 0.216}
\definecolor{yafcolor3}{rgb}{0.47, 0.549, 0.306}
\definecolor{yafcolor4}{rgb}{0.925, 0.165, 0.224}
\definecolor{yafcolor5}{rgb}{0.141, 0.345, 0.643}
\definecolor{yafcolor6}{rgb}{0.965, 0.933, 0.267}
\definecolor{yafcolor7}{rgb}{0.627, 0.118, 0.165}
\definecolor{yafcolor8}{rgb}{0.878, 0.475, 0.686}

\tikzstyle{machnode} = [circle, fill = yafcolor3, text  = white, inner sep = 2pt, minimum width = 13.5pt]
\tikzstyle{machedge} = [draw, thick, >=latex, ->]

\tikzstyle{exnode} = [circle, fill = yafcolor5, text  = white, inner sep = 2pt, minimum width = 13.5pt]
\tikzstyle{exedge} = [draw, thick, >=latex, ->]
\tikzstyle{toynode} = [fill = yafcolor5, circle, inner sep = 1.5pt]
\tikzstyle{toyedge} = [draw = black, >=latex, ->, opacity = 0.5]
\tikzstyle{smallnode} = [inner sep = 1.5pt, font=\scriptsize]
\tikzstyle{mediumnode} = [inner sep = 1.5pt]
\tikzstyle{smalledge} = [draw = black, >=latex, ->]

\newsavebox{\edgebox}
\tikzexternaldisable
\savebox{\edgebox}
{\begin{tikzpicture}[baseline]
\draw (0, 0.06) edge[smalledge] (0.35, 0.06);
\end{tikzpicture}}
\tikzexternalenable

\newcommand{\epito}{\usebox{\edgebox} }
\newcommand{\epitoc}{\usebox{\edgebox}}
\newcommand{\epispace}{\hspace{3mm}}

\newlength{\yafaxispad}
\newlength{\yaftlpad}
\newlength{\yaflabelpad}
\newlength{\yafaxiswidth}
\newlength{\yafticklen}
\makeatletter
\def\pgfplots@drawtickgridlines@INSTALLCLIP@onorientedsurf#1{}
\makeatother

\newcommand{\yafdrawxaxis}[2]{
	\pgfplotstransformcoordinatex{#1}\let\xmincoord=\pgfmathresult 
	\pgfplotstransformcoordinatex{#2}\let\xmaxcoord=\pgfmathresult 
	\pgfsetlinewidth{\yafaxiswidth} 
	\pgfsetcolor{yafaxiscolor}
	\pgfpathmoveto{\pgfpointadd{\pgfpointadd{\pgfplotspointrelaxisxy{0}{0}}{\pgfqpointxy{\xmincoord}{0}}}{\pgfqpoint{-0.5\yafaxiswidth}{\yafaxispad}}}
	\pgfpathlineto{\pgfpointadd{\pgfpointadd{\pgfplotspointrelaxisxy{0}{0}}{\pgfqpointxy{\xmaxcoord}{0}}}{\pgfqpoint{0.5\yafaxiswidth}{\yafaxispad}}}
	\pgfusepath{stroke}

}
\newcommand{\yafdrawyaxis}[2]{
	\pgfplotstransformcoordinatey{#1}\let\ymincoord=\pgfmathresult 
	\pgfplotstransformcoordinatey{#2}\let\ymaxcoord=\pgfmathresult 
	\pgfsetlinewidth{\yafaxiswidth} 
	\pgfsetcolor{yafaxiscolor}
	\pgfpathmoveto{\pgfpointadd{\pgfpointadd{\pgfplotspointrelaxisxy{0}{0}}{\pgfqpointxy{0}{\ymincoord}}}{\pgfqpoint{\yafaxispad}{-0.5\yafaxiswidth}}}
	\pgfpathlineto{\pgfpointadd{\pgfpointadd{\pgfplotspointrelaxisxy{0}{0}}{\pgfqpointxy{0}{\ymaxcoord}}}{\pgfqpoint{\yafaxispad}{0.5\yafaxiswidth}}}
	\pgfusepath{stroke}
}

\newcommand{\yafdrawaxis}[4]{\yafdrawxaxis{#1}{#2}\yafdrawyaxis{#3}{#4}}

\pgfplotscreateplotcyclelist{yaf}{%
{yafcolor1,mark options={scale=0.75},mark=o}, 
{yafcolor2,mark options={scale=0.75},mark=square},
{yafcolor3,mark options={scale=0.75},mark=triangle},
{yafcolor4,mark options={scale=0.75},mark=o},
{yafcolor5,mark options={scale=0.75},mark=o},
{yafcolor6,mark options={scale=0.75},mark=o},
{yafcolor7,mark options={scale=0.75},mark=o},
{yafcolor8,mark options={scale=0.75},mark=o}} 

\pgfplotsset{axis y line=left, axis x line=bottom,
	tick align=outside,
	compat = 1.3,
	tickwidth=\yafticklen,
	clip = false,
	every axis title shift = 0pt,
    x axis line style= {-, line width = 0pt, opacity = 0},
    y axis line style= {-, line width = 0pt, opacity = 0},
    x tick style= {line width = \yafaxiswidth, color=yafaxiscolor, yshift = \yafaxispad},
    y tick style= {line width = \yafaxiswidth, color=yafaxiscolor, xshift = \yafaxispad},
    x tick label style = {font=\scriptsize, yshift = \yaftlpad},
    y tick label style = {font=\scriptsize, xshift = \yaftlpad},
    every axis y label/.style = {at = {(ticklabel cs:0.5)}, rotate=90, anchor=center, font=\scriptsize, yshift = -\yaflabelpad},
    every axis x label/.style = {at = {(ticklabel cs:0.5)}, anchor=center, font=\scriptsize, yshift = \yaflabelpad},
    x tick label style = {font=\scriptsize, yshift = 1pt},
    grid = major,
    major grid style  = {dash pattern = on 1pt off 3 pt},
	every axis plot post/.append style= {line width=\yafaxiswidth} ,
	legend cell align = left,
	legend style = {inner sep = 1pt, cells = {font=\scriptsize}},
	legend image code/.code={%
		\draw[mark repeat=2,mark phase=2,#1] 
		plot coordinates { (0cm,0cm) (0.15cm,0cm) (0.3cm,0cm) };%
	} 
}

\begin{document}

\title{Ranking Episodes using a Partition Model}

\author{Nikolaj Tatti}
\institute{
N. Tatti \at
HIIT, Department of Information and Computer Science, Aalto University, Finland\\
\email{nikolaj.tatti@aalto.fi}}


\maketitle

\begin{abstract}

One of the biggest setbacks in traditional frequent pattern mining is that
overwhelmingly many of the discovered patterns are redundant.
A prototypical
example of such redundancy is a freerider pattern where the pattern 
contains a true pattern and some additional noise events.  A technique for
filtering freerider patterns that has proved to be efficient in ranking
itemsets is to use a partition model where a pattern is divided into two
subpatterns and the observed support is compared to the expected support under
the assumption that these two subpatterns occur independently.

In this paper we develop a partition model for episodes, patterns discovered
from sequential data.  An episode is essentially a set of events,
with possible restrictions on the order of events. Unlike with itemset mining,
computing the expected support of an episode requires surprisingly sophisticated
methods.  In order to construct the model, we partition the episode
into two subepisodes. We then model how likely the events in each subepisode
occur close to each other.  If this probability is high---which is often the
case if the subepisode has a high support---then we can expect that when one
event from a subepisode occurs, then the remaining events occur also close by.
This approach increases the expected support of the episode, and if this
increase explains the observed support, then we can deem the episode
uninteresting.  We demonstrate in our experiments that using the partition
model can effectively and efficiently reduce the redundancy in episodes.

\end{abstract}




\section{Introduction}

Pattern mining is one of the most well-studied subfields in exploratory data
analysis.  One of the major setbacks of traditional frequent pattern mining
techniques is that the obtained results are heavily redundant. 
Hence, the focus of the pattern mining field has moved away from mining patterns
efficiently to reducing redundancy of the output. This has been especially
the case for mining itemsets.

A technique to reduce redundancy that has proved to be efficient for itemsets
is to use a partition model~\citep{webb:10:self-sufficient}. A partition model
for itemsets involves in dividing an itemset, say $Z$, into two subitemsets, say
$X$ and $Y$, and assume that items in $X$ and $Y$ are independent.
If the observed support of $Z$ is close to the expected support, then we deem $Z$
uninteresting.
In order to
select $X$ and $Y$ we simply iterate over all possible partitions and pick the
one that fits the best with the observed data.  For example, if $Z$ is an
itemset that contains an interesting pattern \emph{and} some independent noise
events, then the partition model is able to detect this and downplay the
importance of $Z$.

In this paper our goal is to reduce redundancy in episodes, a very general class of sequential
patterns~\citep{mannila:97:discovery}.  Essentially, an episode is a set of
events that should occur in a sequence. In addition, these events may have
constraints on the order in which they should occur. This order is expressed by
a directed acyclic graph (DAG).

While ranking and filtering patterns to reduce redundancy is well-studied for itemsets, it is surprisingly underdeveloped for episodes. The most
straightforward approach to rank episodes is to compare them against the
independence
model~\citep{gwadera:05:reliable,DBLP:conf/icdm/Low-KamRKP13,tatti:14:mining}.
In this paper we will introduce ranking technique based on partition models
instead of independence model. Our goal is that by using partition models we
will be able to reduce redundancy in episodes in a similar fashion that
partition models allow us to reduce redundancy in
itemsets~\citep{webb:10:self-sufficient}.

Computing the expected support for an episode is a more intricate process than
computing the expected support for an itemset. For example, to obtain the
expected support of an itemset according to the independence model we can
simply multiply the margins of individual items. On the other hand, to compute
the expectation for an episode, we need to construct a finite state machine,
where each state represents the episode events that we have seen so far (see
Section~\ref{sec:independence} for more details). We can then compute the
expected support by computing the probability of a random sequence reaching the
final state of the machine.

We will consider two types of partition models. In the first approach we
partition the episode into two subepisodes. If one or both of these subepisodes
have few gaps, then we will increase the probability of a whole subepisode to
occur in a sequence once we have seen at least one event from the subepisode.
This will increase the expected support of the episode.  In the second approach
we try to explain the support of an episode with an episode that has the same
events but impose more strict constraints on the order. In this case we will increase the
probability of events whenever they obey the more strict order. 

Fortunately, we can construct the partition model for both aforementioned cases
using the same finite state machine that we use to compute the expectation for
the independence model.  Roughly speaking, when computing the probability of
reaching the final state of the finite state machine, we will increase the
probability of a random sequence taking certain transitions. These transitions
will be determined either by the subepisodes (the first case) or by the
superepisode (second case). In both cases,  this will increase the
probability of a random sequence containing the episode and will increase the
expected support of an episode.

The rest paper of the paper is organized as follows. We introduce preliminary
notation in Section~\ref{sec:prel}. In Section~\ref{sec:rank} we describe how to
rank episodes given the model. In Section~\ref{sec:independence} we construct a
finite state machine that we need to compute the independence model.
Our main methodological contribution is given in the next two sections.
In Section~\ref{sec:partmodel} we obtain a partition model by boosting
certain transitions of the finite state machine. We introduce the partition
model using subepisodes and superepisodes in
Section~\ref{sec:select}. We discuss the related work in
Section~\ref{sec:related}. Finally, we introduce our experimental evaluation in
Section~\ref{sec:exp} and conclude the paper with discussion in
Section~\ref{sec:conclusions}.

\section{Preliminaries}\label{sec:prel}

We begin by introducing the notation that we will use throughout the paper.

Our input dataset consists of $m$ sequences $\mathcal{S} = S_1, \ldots, S_m$.
Each sequence contains events coming from some finite universe, which we will
denote by $\Sigma$.

We are interested in episodes introduced by~\citet{mannila:97:discovery} and
defined as follows.

\begin{definition}
An \emph{episode} $G = (V, E, \lab{})$ is a directed acyclic graph with labelled
vertices.  The labels are represented by the label function $\funcdef{\lab{}}{V(G)}{\Sigma}$, mapping
each vertex to a label.
\end{definition}

We will call $G$ a \emph{parallel} episode if $G$ has no edges.
On the other hand, an episode that represents a total order is called a \emph{serial} episode.

Informally, an episode represents a set of events that should occur in the order
that is consistent with the edges. More formally:

\begin{definition}
\label{def:cover}
Given a sequence $S = s_1, \ldots, s_n$ and an episode $G = (V, E, \lab{})$, we say
that $S$ \emph{covers} $G$ if there is an injective mapping $m$ from the vertices of $G$
to the indices of $s$, $\funcdef{m}{V(G)}{1, \ldots, n}$, such that
\begin{enumerate}
\item  labels are honored, $s_{m(v)} = \lab{v}$ for every $v \in V$, 
\item edges are honored, $m(v) < m(w)$ if $(v, w) \in E$.
\end{enumerate}
\end{definition}

\begin{example}
Consider an episode $G_1$ given in Figure~\ref{fig:toy}.
Definition~\ref{def:cover} implies that 
a sequence $S$ covers $G_1$ if and only if $S$ contains $a$
followed by $b$ and $c$ in arbitrary order, and finally followed by $d$, with any number of events
before between, or after these 4 events.
For example, $aebfcd$ covers $G_1$ due to a subsequence $abcd$ 
but $cabde$ does not since there is no $c$ between $a$ and $d$.
\end{example}

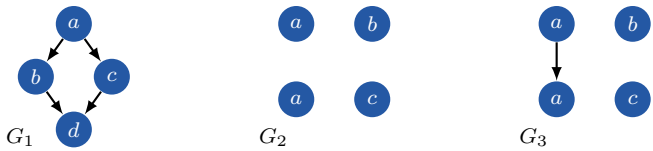
\begin{figure}[ht!]
\hfill
\begin{tikzpicture}
\node[exnode] (a1) {$a$};
\node[exnode] (a2) at (-0.5, -0.7) {$b$};
\node[exnode] (a3) at (0.5, -0.7) {$c$};
\node[exnode] (a4) at (0, -1.4) {$d$};
\draw[exedge] (a1) -- (a2);
\draw[exedge] (a1) -- (a3);
\draw[exedge] (a2) -- (a4);
\draw[exedge] (a3) -- (a4);
\node at (-0.7, -1.5) {$G_1$};
\end{tikzpicture}\hfill
\begin{tikzpicture}
\node[exnode] (a1) {$a$};
\node[exnode] (a2) at (0, -1) {$a$};
\node[exnode] (a3) at (1, 0) {$b$};
\node[exnode] (a4) at (1, -1) {$c$};
\node at (-0.3, -1.5) {$G_2$};
\end{tikzpicture}\hfill
\begin{tikzpicture}
\node[exnode] (a1) {$a$};
\node[exnode] (a2) at (0, -1) {$a$};
\node[exnode] (a3) at (1, 0) {$b$};
\node[exnode] (a4) at (1, -1) {$c$};
\node at (-0.3, -1.5) {$G_3$};
\draw[exedge] (a1) -- (a2);
\end{tikzpicture}
\hspace*{\fill}
\caption{Toy episodes used in examples}
\label{fig:toy}
\end{figure}

Note that an episode and its transitive closure represent essentially the same
pattern: an episode is covered if and only its transitive closure is covered.
For simplicity, we will assume that we are only dealing with transitively closed
episodes. However, for aesthetic reasons, whenever showing an episode we will
remove edges that are implied by the transitive closure.

Now that we have defined occurrence in a single sequence, we can define the 
support of an episode.

\begin{definition}
Let $G$ be an episode and $\mathcal{S} = S_1, \ldots, S_m$ be a collection of sequences.
The \emph{support} of $G$ is the number of sequences covering $G$,
\[
	\supp{G} = \abs{\set{i \mid S_i \text{ covers } G}}\quad.
\]
\end{definition}

Since the support is monotonically decreasing, discovering \emph{all} episodes
whose support is higher than some given threshold can be done efficiently using
\textsc{APriori} or \textsc{DFS} style approach.

Unlike itemsets, episodes are surprisingly difficult to handle. For example,
checking whether a sequence covers an episode is in fact an \np-hard problem~\citep{tatti:11:mining}.

We will focus on a more simple class of episodes, which are called strict
episodes~\citep{tatti:12:mining}.

\begin{definition}
An episode $G = (V, E, \lab{})$ is \emph{strict} if any two
distinct vertices $v, w \in V$ with the same label, $\lab{v} = \lab{w}$ we have
either $(v, w) \in E$ or $(w, v) \in E$.
\end{definition}

The need for using strict episodes stems from technical details that we will see
in later sections. Nevertheless, this class of episodes is large: it
contains all serial episodes, all episodes with unique labels. In addition, for
every parallel episode $G$, there is a strict episode $H$ such that a sequence
will cover $G$ if and only if the same sequence covers $H$. To obtain $H$ from $G$
simply connect all vertices with the same label. For example, $G_2$ in
Figure~\ref{fig:toy} is a parallel episode while $G_3$ is a strict episode, and
a sequence $S$ covers $G_2$ if and only if $S$ covers $G_3$ as well.

From now on we will assume that episodes are strict.

We will need a concept of an induced episode which is essentially a standard notion of
an induced graph.

\begin{definition}
Given an episode $G = (V, E, \lab{})$ and a subset of vertices $W \subseteq V$, we define
an induced episode $G(W)$ to be the episode with vertices $W$ and edges 
\[
	E(W) = \set{(v, w) \in E \mid v, w \in W})\quad.
\]
The vertices have the same labels as the vertices in $G$.
\end{definition}

\section{Ranking episodes based on expectation}\label{sec:rank}

In this section we describe how to rank episodes based on the expected support.
We will give the details for computing the expectation in the latter sections.

Formally,
consider that we are given an episode $G$ and a dataset of sequences
$\mathcal{S}  = S_1, \ldots, S_m$.
Unlike with itemsets we need to take into account the length of individual
sequences as longer sequences have a higher probability to cover
an episode.
Assume that we have a generative model $M$ for a sequence,
that allows us to compute the probability that 
$G$ occurs in a sequence of a certain length, that is, we can compute
\[
	p_k = p(X \text{ covers } G \mid \abs{X} = k, M),
\]
where $X$ is a random sequence of length $k$. We will define different variants
of $M$ in the next sections.

Let $\mathcal{X}$ be $m$ random sequences, each random sequence having the same length
as the input sequence, $\abs{X_i} = \abs{S_i}$. 
If we assume that each sequence in $\mathcal{X}$ is generated independently,
then the expected support of $G$ according to the model is then
\[
	\mu = \sum_{S \in \mathcal{S}} p_{\abs{S}}\quad.
\]
Moreover, we can easily show that the probability that $\supp{G}$ is equal to
$n$ is
\[
	p(\supp{G ; \mathcal{X}} = n \mid M) = \sum_{\mathcal{T} \subseteq \mathcal{S} \atop \abs{\mathcal{T}} = n} \prod_{S \in \mathcal{T}} p_{\abs{S}}\prod_{S \in \mathcal{S} \setminus \mathcal{T}} (1 - p_{\abs{S}}),
\]
where the sum goes over all subsets of $\mathcal{S}$ of size $n$. If all sequences
are of equal length, then this distribution is in fact a binomial distribution.

Assume that we observe the support to be $\supp{G ; \mathcal{S}} = n$.
Ideally, we would like to compute the rank to be the probability $p(\supp{G; \mathcal{X}} \geq n \mid M)$.
This value is close to $1$ whenever support is low and $0$ whenever the support is large.
Note that this quantity can be interpreted as a $p$-value. However, in this work
we will not make this interpretation and treat this quantity simply as a rank (see Section~\ref{sec:conclusions} for discussion about interpreting
this quantity as a $p$-value).
Since in practice most of the values will be very close to $0$ we consider the logarithm of the score, that is,
we define
\[
\begin{split}
	\rank{G \mid M} & = - \log p(\supp{G; \mathcal{X}} \geq n \mid M) \\
	                & = - \log 1 - \sum_{k = 1}^{n - 1} p(\supp{G; \mathcal{X}} = k \mid M) \quad.
\end{split}
\]
Episodes that have abnormally high support will have a high rank.  Computing
the rank can be done in $O(n^2m)$ with a simple recursive equation. However,
this may be slow if $n$, the observed support, is large. Hence, in practice we
will use well-known asymptotic estimates for $p(\supp{G; \mathcal{X}} \geq n \mid M)$.
If $n$ is large enough, we can estimate the probability with a normal
distribution $N(\mu, \sigma)$, where the variance $\sigma^2$ is 
\[
	\sigma^2 = \sum_{S \in \mathcal{S}} p_{\abs{S}}(1 - {p_\abs{S}})\quad.
\]
In practice, the input dataset is large
enough so that the approximation is accurate \emph{if} $\mu$ is not close to $0$.
If $\mu$ is small, say $\mu \leq 10$, this approximation becomes inaccurate.
In such cases, a common approach is to estimate the
probability with Poisson distribution with a mean of $\mu$.

\section{Independence model for episodes}\label{sec:independence}
In this section we review how to compute the expected support of an episode
using the independence model. The idea of computing the expected support
using the independence model was originally done by~\citet{gwadera:05:reliable}.
This approach requires us to construct a certain
finite state machine.  In later sections we will use this machine to build
the partition model.

\subsection{Finding episodes with finite state machine}

Our first goal is to construct a finite state machine from an episode.  This
machine has two purposes. Firstly, we can use it to compute the support of an
episode. Secondly, we can use it to compute the expected support, either using the independence
model, which we will review in Section~\ref{sec:indmodel} or the partition model
which we will introduce in Section~\ref{sec:partmodel}.

We start with a definition of a prefix graph which will turn out to be the states
of our machine.

\begin{definition}
Given an episode $G = (V, E, \lab{})$ and a subset of vertices $W \subseteq V$,
we say that an induced graph $H = G(W)$ is a \emph{prefix subgraph} if all ancestors of vertices in $W$ are also
included in $W$, that is,
\[
	v \in W \text{ and } (w, v) \in E \quad\text{implies}\quad w \in W\quad.
\]
We will denote the collection of all prefix graphs by $\pre{G}$.
\end{definition}

\begin{example}
Consider an episode $G$ given in Figure~\ref{fig:toyprefix}.
This episode has 6 prefix graphs $H_1, \ldots, H_6$, given also in
Figure~\ref{fig:toyprefix}. Note that the empty graph $G(\emptyset)$
and the full graph $G$ are both prefix graphs.
\end{example}

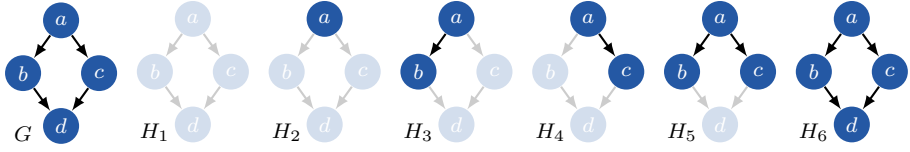
\begin{figure}[ht!]
\begin{center}
\begin{tikzpicture}
\node[exnode] (a1) {$a$};
\node[exnode] (a2) at (-0.5, -0.7) {$b$};
\node[exnode] (a3) at (0.5, -0.7) {$c$};
\node[exnode] (a4) at (0, -1.4) {$d$};
\draw[exedge] (a1) -- (a2);
\draw[exedge] (a1) -- (a3);
\draw[exedge] (a2) -- (a4);
\draw[exedge] (a3) -- (a4);
\node at (-0.5, -1.5) {$G$};
\end{tikzpicture}\hfill
\begin{tikzpicture}
\node[exnode, fill = yafcolor5!20] (a1) {$a$};
\node[exnode, fill = yafcolor5!20] (a2) at (-0.5, -0.7) {$b$};
\node[exnode, fill = yafcolor5!20] (a3) at (0.5, -0.7) {$c$};
\node[exnode, fill = yafcolor5!20] (a4) at (0, -1.4) {$d$};
\draw[exedge, black!20] (a1) -- (a2);
\draw[exedge, black!20] (a1) -- (a3);
\draw[exedge, black!20] (a2) -- (a4);
\draw[exedge, black!20] (a3) -- (a4);
\node at (-0.5, -1.5) {$H_1$};
\end{tikzpicture}\hfill
\begin{tikzpicture}
\node[exnode] (a1) {$a$};
\node[exnode, fill = yafcolor5!20] (a2) at (-0.5, -0.7) {$b$};
\node[exnode, fill = yafcolor5!20] (a3) at (0.5, -0.7) {$c$};
\node[exnode, fill = yafcolor5!20] (a4) at (0, -1.4) {$d$};
\draw[exedge, black!20] (a1) -- (a2);
\draw[exedge, black!20] (a1) -- (a3);
\draw[exedge, black!20] (a2) -- (a4);
\draw[exedge, black!20] (a3) -- (a4);
\node at (-0.5, -1.5) {$H_2$};
\end{tikzpicture}\hfill
\begin{tikzpicture}
\node[exnode] (a1) {$a$};
\node[exnode] (a2) at (-0.5, -0.7) {$b$};
\node[exnode, fill = yafcolor5!20] (a3) at (0.5, -0.7) {$c$};
\node[exnode, fill = yafcolor5!20] (a4) at (0, -1.4) {$d$};
\draw[exedge] (a1) -- (a2);
\draw[exedge, black!20] (a1) -- (a3);
\draw[exedge, black!20] (a2) -- (a4);
\draw[exedge, black!20] (a3) -- (a4);
\node at (-0.5, -1.5) {$H_3$};
\end{tikzpicture}\hfill
\begin{tikzpicture}
\node[exnode] (a1) {$a$};
\node[exnode, fill = yafcolor5!20] (a2) at (-0.5, -0.7) {$b$};
\node[exnode] (a3) at (0.5, -0.7) {$c$};
\node[exnode, fill = yafcolor5!20] (a4) at (0, -1.4) {$d$};
\draw[exedge, black!20] (a1) -- (a2);
\draw[exedge] (a1) -- (a3);
\draw[exedge, black!20] (a2) -- (a4);
\draw[exedge, black!20] (a3) -- (a4);
\node at (-0.5, -1.5) {$H_4$};
\end{tikzpicture}\hfill
\begin{tikzpicture}
\node[exnode] (a1) {$a$};
\node[exnode] (a2) at (-0.5, -0.7) {$b$};
\node[exnode] (a3) at (0.5, -0.7) {$c$};
\node[exnode, fill = yafcolor5!20] (a4) at (0, -1.4) {$d$};
\draw[exedge] (a1) -- (a2);
\draw[exedge] (a1) -- (a3);
\draw[exedge, black!20] (a2) -- (a4);
\draw[exedge, black!20] (a3) -- (a4);
\node at (-0.5, -1.5) {$H_5$};
\end{tikzpicture}\hfill
\begin{tikzpicture}
\node[exnode] (a1) {$a$};
\node[exnode] (a2) at (-0.5, -0.7) {$b$};
\node[exnode] (a3) at (0.5, -0.7) {$c$};
\node[exnode] (a4) at (0, -1.4) {$d$};
\draw[exedge] (a1) -- (a2);
\draw[exedge] (a1) -- (a3);
\draw[exedge] (a2) -- (a4);
\draw[exedge] (a3) -- (a4);
\node at (-0.5, -1.5) {$H_6$};
\end{tikzpicture}
\end{center}
\caption{Toy episode $G$ and all the prefix graphs $H_1, \ldots, H_6$.}
\label{fig:toyprefix}
\end{figure}

Now that we have defined the states of our machine, we can finally define the machine itself.

\begin{definition}
Given an episode $G$ we define a \emph{machine} $\mach{G}$ to be a DAG with labelled edges,
such that the states are the prefix subgraphs $\pre{G}$ and two states $H_1$
and $H_2$ are connected with an edge $(H_1, H_2)$ if we can obtain $H_1$ by
deleting a (sink) vertex from $H_2$. The label of the edge is the label of the deleted vertex.

The \emph{source state} of $\mach{G}$ is the empty prefix graph $G(\emptyset)$, while
the \emph{sink state} is the episode $G$ itself.
\end{definition}

\begin{example}
Consider an episode $G$ given in Figure~\ref{fig:toyprefix}. This episode has
$6$ prefix graphs given also in Figure~\ref{fig:toyprefix} and the machine
$\mach{G}$ is given in Figure~\ref{fig:toymachine}.
\end{example}

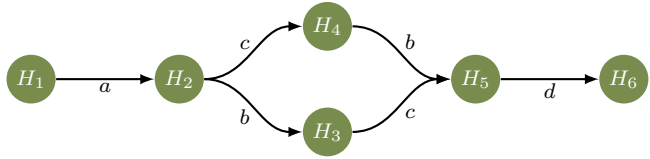
\begin{figure}[ht!]
\begin{center}
\begin{tikzpicture}[auto]
\node[machnode] (n1) {$H_1$};
\node[machnode] at (7.8*0.25, 0) (n2) {$H_2$};
\node[machnode] at (7.8*0.5, -0.7) (n3) {$H_3$};
\node[machnode] at (7.8*0.5, 0.7) (n4) {$H_4$};
\node[machnode] at (7.8*0.75, 0) (n5) {$H_5$};
\node[machnode] at (7.8, 0) (n6) {$H_6$};

\draw[machedge, in = 180, out = 0] (n1) edge node[swap, inner sep = 1pt] {$a$} (n2);
\draw[machedge, in = 180, out = 0] (n2) edge node[swap, inner sep = 1pt] {$b$} (n3);
\draw[machedge, in = 180, out = 0] (n2) edge node[inner sep = 1pt] {$c$} (n4);
\draw[machedge, in = 180, out = 0] (n3) edge node[swap, inner sep = 1pt] {$c$} (n5);
\draw[machedge, in = 180, out = 0] (n4) edge node[inner sep = 1pt] {$b$} (n5);
\draw[machedge, in = 180, out = 0] (n5) edge node[swap, inner sep = 1pt] {$d$} (n6);

\end{tikzpicture}
\end{center}
\caption{A machine $\mach{G}$ for the episode given in Figure~\ref{fig:toyprefix}.}
\label{fig:toymachine}
\end{figure}

Note that we can view $\mach{G}$ as a finite state machine with a small
technical difference. Finite state machine requires that we should specify transitions
from each state for \emph{every possible} label. We can think of $\mach{G}$ as
a finite state machine by adding self-loops for every possible missing label.
However, it is more natural to ignore these self-loops from the notation since
in practice $\mach{G}$ is implemented as a DAG.

Our next technical lemma is the key result why we are working only with strict
episodes. We will see later on how this lemma helps us with the definitions and
propositions.

\begin{lemma}
\label{lem:strict}
Let $G$ be a strict episode and let $H$ be a state in $\mach{G}$.
Each outgoing edge from $H$ has a unique label among outgoing edges.
Each incoming edge to $H$ has a unique label among incoming edges.
\end{lemma}

\begin{proof}
Assume that there are two edges $(H, F_1)$ and $(H, F_2)$ having the same
label.  This means that there are two distinct vertices $v$ and $w$ in $G$ with
the same label such that $H = F_1 \setminus v$ and $H = F_2 \setminus w$.
Since $G$ is strict, $v$ and $w$ must be connected. Assume that $(v, w) \in E(G)$.
This means that $F_2$ cannot be a prefix graph since $v$ is a parent of $w$
and is not in $V(F_2)$. This is a contradiction and shows that every outgoing edge
has a unique label. The proof for incoming edges is similar.
\qed
\end{proof}

Our next definition is a greedy function mapping a sequence and an initial
state to a final state. The final state is essentially a state that we will end
up by walking greedily the edges $\mach{G}$.

\begin{definition}
Given a machine $M = \mach{G}$ for a strict episode $G$, a state $H$ in $M$, and a sequence $S = s_1, \ldots, s_n$, we
define $\greedy{M, S, H}$ to be the state to which $s$ leads $M$ from $H$,
that is, we can define $\greedy{M, S, H}$ recursively by first defining $\greedy{}$
for the empty sequence, $\greedy{M, \emptyset, H} = H$, and then more generally,
for $i = 1, \ldots, n$,
\[
	\greedy{M, s_i, \ldots, s_n, H} = \greedy{M, s_{i + 1}, \ldots, s_n, F}
\]
if $(H, F) \in E(M)$ with a label $s_i$, and
\[
	\greedy{M, s_i, \ldots, s_n, H} = \greedy{M, s_{i + 1}, \ldots, s_n, H},
\]
otherwise.
\end{definition}

We will abbreviate $\greedy{M, S, G(\emptyset)}$ by $\greedy{M, S}$.

Note that this definition is only well defined if $\mach{G}$ has unique
outgoing edges.  Lemma~\ref{lem:strict} guarantees this since $G$ is a strict
episode. 

\begin{example}
Consider $M = \mach{G}$ given in Figure~\ref{fig:toymachine}. Then, for example,
\[
	\greedy{M, adc} = \greedy{M, adc, H_1} = H_4\quad\text{and}\quad
	\greedy{M, bcd, H_2} = H_6\quad.
\]
\end{example}

One of the reasons we defined $\mach{G}$ is the fact that we can use this to
detect when a sequence is covering $G$. First let us define the
coverage for a \emph{state} in $\mach{G}$.

\begin{definition}
We say that a sequence $S$ covers a state $H$ in $\mach{G}$ if there is a
subsequence $T$ of $S$ leading from the source state to $H$, that is,
\[
	\greedy{M, T} = H\quad.
\]
\end{definition}

As expected, covering an episode $G$ and the sink state in $\mach{G}$ are
closely related. 

\begin{proposition}[Proposition 1 in~\citep{tatti:14:mining}]
\label{prop:subsequence}
Sequence $S$ covers an episode $G$ if and only if $S$ covers the sink state in $\mach{G}$.
\end{proposition}

The technical difficulty with using the definition of coverage is that we
need to \emph{find a subsequence} that travels from the source state to the
sink state.  Fortunately, the next result states that we can simply use the
whole sequence.

\begin{proposition}[Corollary 1 in~\citep{tatti:14:mining}]
\label{prop:greedy}
Sequence $S$ covers the sink state in $\mach{G}$ if and only if $\greedy{M, S} = G$.
\end{proposition}

For the sake of completeness we provide the proof in the appendix.

\begin{example}
Consider $G$ given in Figure~\ref{fig:toyprefix} and its corresponding machine $M =
\mach{G}$ given in Figure~\ref{fig:toymachine}. Sequence $S = aebfcd$ covers
$G$. Proposition~\ref{prop:subsequence} implies that there is a subsequence of $S$,
say $T$, such that $\greedy{M, T} = H_6$ and Proposition~\ref{prop:greedy} makes
a stronger claim that one can choose $T = S$. By applying the definition of the greedy
function, we can easily verify that indeed $\greedy{M, S} = H_6$.
\end{example}

We should point out that this does not hold for a general finite state machine,
however, this holds for any $\mach{G}$.

\subsection{Independence model}\label{sec:indmodel}

Our next step is to compute the expected support.  Here we use the results from
the previous section, by computing the probability that a random sequence
reaches the sink state.

We will use the following notation.

\begin{definition}
Let $M = \mach{G}$ be a machine and let $H$ be a state.
Let $S = s_1, \ldots, s_n$ be a random sequence of $n$ events, generated independently.
Define
\[
	\pind{H, n} = p(\greedy{M, S} = H)
\]
to be the probability that $S$ leads to $H$ from the source state.
\end{definition}

In other words, the probability that a sequence of $n$ events covers $G$ is
equal to $\pind{G, n}$.

We can now compute the probability recursively using the following proposition.

\begin{proposition}
\label{prop:indcascade}
Let $M = \mach{G}$ be a machine and let $H$ be a state.
Let $S = s_1, \ldots, s_n$ be a random sequence of $n$ events, generated independently.
Then the probability of $\greedy{M, S} = H$ is equal to
\[
	\pind{H, n} =  q \times \pind{H, n - 1}  
	+ \sum_{e = (F, H) \in E(M)} p(\lab{e}) \pind{F, n - 1},
\]
where $q$ is the probability of being stuck in $H$ for a single event
\[
	q = 1 - \sum_{e = (H, F) \in E(M)} p(\lab{e}) \quad.
\]
\end{proposition}

This proposition is a special case of Proposition~\ref{prop:partcascade},
hence we will omit the proof. 

If we write $M_{\mathit{ind}}$ to be the independence model, we define
\[
	\rind{G} = \rank{G \mid M_{\mathit{ind}}}\quad.
\]
To compute this rank we need to compute the probability that a random sequence
of length $k$ covers episode $G$. This is exactly what Proposition~\ref{prop:indcascade} does.

\begin{example}
Assume that the alphabet consists of 5 labels and 
the probabilities for labels are $p(a) = 0.4$, $p(b) = 0.3$, $p(c) = 0.2$, $p(d) = 0.06$, and $p(e) = 0.04$.
Consider $M$ given in Figure~\ref{fig:toymachine}.
The initial probabilities are 
\[
	\pind{H_1, 0} = 1, \quad\pind{H_j, 0} = 0, \quad\text{for}\quad j = 2, \ldots, 6\quad.
\]
According to Proposition~\ref{prop:indcascade} the probabilities are
\[
\begin{split}
	\pind{H_1, n + 1} & = 0.6\pind{H_1, n},\\
	\pind{H_2, n + 1} & = 0.5 \pind{H_2, n} + 0.4\pind{H_1, n}, \\
	\pind{H_3, n + 1} & = 0.8 \pind{H_3, n} + 0.3\pind{H_2, n},  \\
	\pind{H_4, n + 1} & = 0.7 \pind{H_4, n} + 0.2\pind{H_2, n},\\
	\pind{H_5, n + 1} & = 0.94 \pind{H_5, n} + 0.2\pind{H_3, n} + 0.3\pind{H_4, n}, \\
	\pind{H_6, n + 1} & = 0.06\pind{H_5, n} + \pind{H_6, n}\quad.
\end{split}
\]
\end{example}

\section{Partition model for episodes}\label{sec:partmodel}

Consider an episode $G$ given in Figure~\ref{fig:toyboost} and its machine
$\mach{G}$. Assume that $b$ has tendency to occur soon after $a$ but $c$ is a
freerider: its occurrence is independent of vicinity of $a$ and $b$.
This episode will have a high rank because its support is higher than what
independence model predicts. The reason for this is that $b$ occurs more often
than expected after $a$, that is, we will move sooner from state $H_2$ to $H_3$
and from $H_5$ to $H_6$ sooner than expected.

\begin{figure}[ht!]
\hfill
\begin{tikzpicture}
\node[exnode] (a1) {$a$};
\node[exnode] (a2) at (0, -1) {$b$};
\node[exnode] (a3) at (1, 0) {$c$};
\node at (-0.3, -1.5) {$G$};
\draw[exedge] (a1) -- (a2);
\end{tikzpicture}
\hfill
\begin{tikzpicture}[auto]
\node at (0, -1.2) {$\mach{G}$};
\node[machnode] (n1) {$H_1$};
\node[machnode] at (7.8*0.25, 0) (n2) {$H_2$};
\node[machnode] at (7.8*0.5, 0) (n3) {$H_3$};

\node[machnode] at (7.8*0.125, -1.2) (n4) {$H_4$};
\node[machnode] at (7.8*0.375, -1.2) (n5) {$H_5$};
\node[machnode] at (7.8*0.625, -1.2) (n6) {$H_6$};

\draw[machedge] (n1) edge node[inner sep = 1pt] {$a$} (n2);
\draw[machedge] (n2) edge node[inner sep = 1pt] {$b$} (n3);

\draw[machedge] (n1) edge node[inner sep = 1pt] {$c$} (n4);
\draw[machedge] (n2) edge node[inner sep = 1pt] {$c$} (n5);
\draw[machedge] (n3) edge node[inner sep = 1pt] {$c$} (n6);

\draw[machedge] (n4) edge node[swap, inner sep = 1pt] {$a$} (n5);
\draw[machedge] (n5) edge node[swap, inner sep = 1pt] {$b$} (n6);
\end{tikzpicture}\hspace*{\fill}
\caption{Toy episode and its machine}
\label{fig:toyboost}
\end{figure}
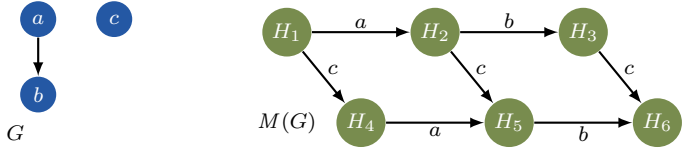

Our goal is to construct a more flexible model that would take into account
that some of the transitions are more probable than what the independence model
predicts. This will allow us to remove the freeriders.

In order to do that let us fix an episode $G$ and assume that we are given two
disjoint subsets of edges $C_1 \subset E(\mach{G})$ and $C_2 \subset
E(\mach{G})$. Note that (both of) these sets can be empty and it is possible that $C_1 \cup C_2 \neq  E(\mach{G})$.
We will describe later on how we select these sets
but for now we will assume that they are given. Also, we can easily define this model for $k$
sets but we only need two sets.

Our model has $\abs{\Sigma} + 2$ parameters: $\abs{\Sigma}$ parameters $u_l$ states
the likelihood of a label $l$. The larger $u_l$, the more likely $l$ is
to occur in a sequence. In addition, we have two transition parameters.
Parameter $t_1$ states how likely we use an edge in $C_1$ while
$t_2$ states how likely we use an edge in $C_2$. 

In order to define the actual model let us first define the conditional probability
of a label \emph{given} a state $H$, 
\[
	p(l \mid H) =
\begin{cases}
\frac{1}{Z_H}\exp\pr{u_l + t_1}, & \text{if there is } (H, F) \in C_1 \text{ and } \lab{(H, F)} = l, \\
\frac{1}{Z_H}\exp\pr{u_l + t_2}, & \text{if there is } (H, F) \in C_2 \text{ and } \lab{(H, F)} = l, \\
\frac{1}{Z_H}\exp\pr{u_l}, &  \text{otherwise},
\end{cases}
\]
where $Z_H$ is a normalization constant guaranteeing that $\sum_l p(l \mid H) = 1$.
Note that $Z_H$ depends on $H$ while $u_l$ and $t_1$ and $t_2$ do not.

This probability implies that the labels with large $u_l$ are more likely to occur.
Moreover, if there is an edge $(H, F) \in C_1$, then the probability of generating
the label of the edge is increased due to $t_1$ (and similarly for $C_2$).

Note that this is well defined because Lemma~\ref{lem:strict} states that labels
for outgoing edges are unique.

We select $u_l$ and $t_i$ by optimizing the likelihood of a sequence.
In order to do this, we first new to define the probability of a sequence. Let us first decompose it
into conditional probabilities,
\[
	p(S) = \prod_{i = 1}^n p(s_i \mid s_1, \ldots, s_{i - 1})\quad.
\]
We define the probability of $s_i$ to be
\[
	p(s_i \mid s_1, \ldots, s_{i - 1}) = p(s_i \mid H),
\]
where $H$ is the state given by the greedy function, 
\[
	H = \greedy{M, (s_1, \ldots, s_{i - 1})}\quad.
\]
In other words, $s_i$ is generated from $p(\cdot \mid H)$, where $H$ is the
current state led by $s_1, \ldots, s_{i - 1}$.

Note that if $C_1 = C_2 = \emptyset$, then $p(l \mid H) = p(l)$, and the model
is in fact the independence model. However, if $C_1$ and $C_2$ are not empty,
certain labels are expected to occur more often\footnote{or more rarely if $t_i$ are
small.} depending on the current state of $\mach{G}$.

Our next step is to compute the probability of a sequence covering an episode.
To that end,
let us define $\pprt{H, n}$ to the probability
according to the partition model that a random sequence of length $n$ reaches $H$.
The following proposition, a
generalization of Proposition~\ref{prop:indcascade}, allows to compute the
expected support.

\begin{proposition}
\label{prop:partcascade}
Let $M = \mach{G}$ be a machine and let $H$ be a state.
Let $S = s_1, \ldots, s_n$ be a random sequence of $n$ events, generated independently.
Then the probability of $\greedy{M, S} = H$ is equal to
\[
	\pprt{H, n} = q \times \pprt{H, n - 1}  + \sum_{e = (F, H) \in E(M)} p(\lab{e} \mid F) \pprt{F, n - 1},
\]
where $q$ is the probability of being stuck in $H$ for a single event
\[
	q = 1 - \sum_{e = (H, F) \in E(M)} p(\lab{e} \mid H) \quad.
\]
\end{proposition}

The proof of this proposition is given in the appendix.

Our final step is to find the parameters $\set{u_i}$, $t_1$, and $t_2$ of the model.
Here we select the parameters optimizing the likelihood of $\mathcal{S}$
\[
	p(\mathcal{S}) = \prod_{i = 1}^m p(S_i),
\]
that is we assume that each sequence in $\mathcal{S}$ is generated independently.
Unlike with the independence model we do not
have a closed solution. However, we can show that the likelihood is a concave
function of $u_l$, $t_1$ and $t_2$.

\begin{proposition}
\label{prop:concave}
$\log p(\mathcal{S})$ is a concave function of the model parameters $\set{u_i}$, $t_1$ and $t_2$.
\end{proposition}

The proof of this proposition is given in the appendix.

The concavity allows us to use gradient methods to find the local maximum which
is guaranteed to be also the global maximum. We used Newton-Raphson method to
find the optimal solution. The technical details for computing the descent are
given in Appendix~\ref{sec:descent}.

\begin{example}
Consider a serial episode $G = a \to b \to c \to x$. Assume that there are
no gap events between $a$ and $b$, and $b$ and $c$, and $x$ occurs
independently of other events.

In such case, the independence model will overestimate the the sizes of gaps
between $a$ and $b$, and $b$ and $c$. This leads to underestimating the
probability $G$ occurring in a sequence of a given length, which ultimately
leads to underestimating the support of $G$.

On the other hand, let us set $C_1 = \set{(a, a \to b), (a \to b, a \to b \to
c)}$ and $C_2 = \emptyset$.  Then the maximum likelihood solution will have
$t_1 = \infty$. This means that the partition model will never generate a gap
event between $a$, $b$, and $c$, which implies an increase in the expected
support.  In fact, since we $x$ assume that $x$ is independent of $a$, $b$, and
$c$, the partition model corresponds exactly the generating model, and
consequently the estimate of the support is unbiased.
\end{example}

\section{Which partition models to use?}\label{sec:select}

Now that we have defined our model for ranking episodes, our next step is to consider
which models to use. That is, how to select $C_1$ and $C_2$. Here we consider two approaches.
In the first approach we consider a partition model rising from a prefix graph and in the
second approach we consider a model rising from a superepisode. Finally, we combine both
of these approaches in Section~\ref{sec:partrank} by selecting the model providing the
best explanation for the support.

\subsection{Partition model from prefix graphs}\label{sec:sub}

Now that we have defined our model, our next step is to select which
transitions in $\mach{G}$ we should boost, that is, how to select
$C_1$ and $C_2$.

We consider two approaches. The first approach, described in this section, is
to divide the episode into two subepisodes. The second that is based on considering
superepisodes will be described in the next section.

Informally, our idea is to consider a prefix graph $H$ of $G$. Every vertex
in $H$ corresponds to possibly several edges in $M(G)$. This will give us the first set
of edges $C_1$. These transitions determine
the occurrence of $H$ in a sequence. The other set of edges, $C_2$, is given by the vertices 
outside $H$.

In order to define this formally, let us first define the set of edges in
$\mach{G}$ based on a subset of vertices.

\begin{definition}
Given an episode $G = (V, E, \lab{})$ and a subset of vertices $W$, define
a subset of edges $\blockp{W \mid G}$ of a machine $M = \mach{G}$,
\[
	\blockp{W \mid G} =  \{(H, F) \in E(M)  \mid  V(H) \cap W \neq \emptyset, V(F) \setminus V(H) \subseteq W \},
\]
that is, $\blockp{W \mid G}$ contains the edges $(H, F)$ such that
\begin{enumerate}
\item $F$ is obtained from $H$ by adding a vertex from $W$,
\item $H$ contains at least one vertex from $W$.\!\footnote{Consequently, $F$ contains at least two vertices from $W$.}
\end{enumerate}
\end{definition}

Let $G = (V, E, \lab{})$ be an episode.  Given a prefix graph $H$ with a vertex
set $W$, we define two sets of edges as $C_1 = \blockp{W \mid G}$ and $C_2
= \blockp{V \setminus W \mid G}$. Since our goal is to explain the support
of $G$ using \emph{smaller} episodes, we will require that $W \neq \emptyset$ and
that $W \neq V$. 

\begin{example}
Consider an episode $G$ given in Figure~\ref{fig:toyprefix} along with its prefix
graphs, and also the corresponding $\mach{G}$ given in Figure~\ref{fig:toymachine}.
There are four possible prefix graphs $H_2, \ldots, H_5$. These graphs give a
rise to the edge sets,
\[
\begin{split}
	H_2: \quad & C_1 = \emptyset, \\
	           & C_2 = \set{(H_4, H_5), (H_3, H_5), (H_5, H_6)}, \\[1mm]
	H_3: \quad & C_1 = \set{(H_2, H_3), (H_4, H_5)}, \\
	           & C_2 = \set{(H_5, H_6)}, \\[1mm]
	H_4: \quad & C_1 = \set{(H_2, H_4), (H_3, H_5)}, \\
	           & C_2 = \set{(H_5, H_6)}, \\[1mm]
	H_5: \quad & C_1 = \set{(H_2, H_3), (H_2, H_4), (H_4, H_5), (H_3, H_5)}, \\
	           & C_2 = \emptyset\quad. \\
\end{split}
\]
Let us consider $H_5$, an episode, where $a$ is followed by $b$ and $c$, in any
order. Let us assume $b$ and $c$ occurs almost immediately after $a$, in other words,
the edges in $C_1$ should be traversed quickly, which leads to a large $t_1$, and elevated
expected support. On the other hand, $(H_5, H_6)$ is not boosted in anyway, that is, we model 
$d$ independently of $a$, $b$, and $c$.

Let us now take a closer look on $H_3 = a \to b$. Assume that $H_3$ has elevated
support and the main reason for this elevated support is that $b$ occurs often
after $a$ almost immediately. Consider now the corresponding edges $C_1$ in $\mach{G}$,
$(H_2, H_3)$ and $(H_4, H_5)$. These edges correspond to seeing $b$ after we have witnessed
$a$ (in the latter we have also witnessed $c$ as a gap event). Hence, by our assumption
these edges should be traversed quickly, that is, $t_1$ should be large.
Similarly, $C_2$ corresponds to $c \to d$,
and if $d$ occurs often $c$, then $t_2$ should also be large.
Consequently, if $t_1$ and/or $t_2$ is large, then the model will yield an
increased expected support for $G$.
\end{example}

Note that in the definition of $\blockp{W \mid G}$ we require that the parent
node of an edge must be a state containing at least one member in $W$. 
For example, outgoing edges of the source state of $\mach{G}$ will never be
a part of $C_1$ or $C_2$.
The idea behind this constraint is that $C_1$ and $C_2$ should not model the likelihood
of finding the first vertex of the prefix graph (or the postfix graph). Instead
we want model how likely we will find the remaining vertices of an episode
once the first vertex is found in a sequence.

To justify the definition of $\blockp{W \mid G}$, consider a machine $N =
\mach{G(W)}$.  An abnormally large support of $G(W)$ suggests that the edges in
$N$ are traversed abnormally fast, that is, the number of gap events is low.
As the following proposition states the edges in $\blockp{W \mid G}$ have a
direct correspondence to the edges in $N$, and so they will be traversed
abnormally fast. By modelling this phenomenon with a parameter $t_1$, we hope
to take into account the large support of $G(W)$.

\begin{proposition}
Let $G = (V, E, \lab{})$ be an episode. Let $W \subseteq V$ be a subset of
vertices such that $G(W)$ or $G(V\setminus W)$ is a prefix subgraph.  Let $M =
\mach{G}$ and $N = \mach{G(W)}$. Define a mapping $\rho$ from states of $M$ to
states of $N$ to be $\rho(H) = H(V(H) \cap W)$. Then $\rho$ is a surjection and 
for any edge in $(H, F) \in E(M)$ one of the following holds
\begin{enumerate}
\item $\rho(H) = \rho(F)$ or $\rho(H)$ is the initial state or
\item $(\rho(H), \rho(F))$ is an edge in $N$ and $(H, F) \in \blockp{W \mid G}$.
\end{enumerate}
In addition, for every edge $(H', F') \in N$ there is an edge $(H, F) \in \blockp{W \mid G}$
such that $H' = \rho(H)$ and $F' = \rho(F)$.
\end{proposition}

\begin{proof}
Assume that $G(W)$ is a prefix subgraph (the $G(V\setminus W)$ case is
similar).  A state $H$ in $N$ corresponds to a prefix subgraph of $G(W)$, which
makes $H$ also a prefix graph of $G$, and, by definition, a state in $M$. The
fact that $\rho(H) = H$, makes $\rho$ a surjection.

Let $H$ and $F$ be two states in $M$ such that $(H, F) \in E(M)$. Then $F$ is
obtained from $H$ by adding one vertex, say $w$. If $w \notin W$, then $\rho(H) = \rho(F)$.
Assume that $w \in W$ and $\rho(H)$ is not the initial state, that is, $H \cap W \neq \emptyset$.
Then, by definition, $(\rho(H), \rho(F))$ is an edge in $N$ and $(H, F) \in \blockp{W \mid G}$

The last statement follows immediately from the fact that $\rho(H) = H$ whenever
$H$ is a prefix subgraph of $G(W)$.\qed
\end{proof}

Let $H$ be a prefix graph of $G$ and let $C_1$ and $C_2$ be the edges
as constructed above. Write $M(C_1, C_2)$ to be the partition model.
We define the rank to be
\[
    \rpart{G; H} = \rank{G \mid M(C_1, C_2)}\quad.
\]
This rank can be computed using Proposition~\ref{prop:partcascade}.
In our experiments, we mimic approach by~\citet{webb:10:self-sufficient} for
itemsets and use the smallest rank among all possible prefix graphs, see
Section~\ref{sec:partrank} for more details.

We should point out that from technical point of view, $H$ in $\rpart{G; H}$
does not need to be a prefix graph. However, models that are generated from
non-prefix graphs may behave unexpectedly.

\begin{example}
\label{ex:nonprefix}
Consider an episode $G = a \to b \to c$ and let $W = \set{a, c}$.
The machine $\mach{G}$ consists of 4 states $H_1 \to H_2 \to H_3 \to H_4$, and $C_1 = \blockp{W \mid G} = (H_3, H_4)$ and $C_2 = \emptyset$.
Note that in this case $t_1$ does \emph{not} model the number of gaps between $a$ and $c$, instead
it models the number of gaps between $b$ and $c$.  In fact, if we set $W' = \set{b, c}$, then $C_1 = \blockp{W' \mid G}$
and $C_2 = \blockp{\set{a} \mid G}$.
\end{example}

The essential problem shown in the example is that there is no direct
transition in $\mach{G}$ of observing $c$ after we have seen $a$.  The
following proposition shows that this problem can be prevented if and only if we use
prefix graphs.

\begin{proposition}
\label{prop:nonprefix}
Let $G$ be an episode and let $M = \mach{G}$ be the corresponding machine. 
Let $W_1$ be a subset of nodes, and let $W_2 = V \setminus W_1$. Then the following statements are equivalent:
\begin{enumerate}
\item either $W_1$ or $W_2$ induces a prefix graph.
\item for any $X \in \pre{G}$ and $i = 1,2$ such that $\emptyset \neq V(X) \cap W_i \neq W_i$ there exists 
$Y \in \pre{G}$, depending on $X$ and $i$, such that $(X, Y) \in \blockp{W_i  \mid G}$.
\end{enumerate}
\end{proposition}

\begin{proof}
The direction \emph{(1)} $\to$ \emph{(2)} is trivial.
Let us prove the other direction. Assume that neither $W_1$ nor $W_2$ induce a
prefix graph. Then there is $v \in W_2$ and $w \in W_1$ such that $(v, w) \in E$.
Let $X$ be the largest prefix graph not containing $v$, such graph exists as the union of the two
prefix graphs is a prefix graph.

Assume that $V(X) \cap W_1 \neq \emptyset$.
Since $w \notin V(X)$, we also have $V(X) \cap W_1 \neq W_1$. Assume that that there
is $Y \in \pre{G}$, such that $(X, Y) \in \blockp{W_1  \mid G}$. By definition, $Y$ is obtained
from $X$ by adding a vertex from $W_1$. Since $Y$ also does not contain $v$, this violates the maximality of $X$.

Assume that $V(X) \cap W_1 = \emptyset$. This is possible only if $(v, u) \in E$ for every $u \in W_1$.
Since $W_2$ does not induce a prefix graph, there is $a \in W_1$ and $b \in W_2$ such that $(a, b) \in E$.
Define $X'$ to be the maximal prefix graph not containing $a$. This graph contains $v$ and does
not contain $b$. Consequently,  $W_2 \neq V(X') \cap W_2 \neq \emptyset$. Assume that that there
is $Y \in \pre{G}$, such that $(X', Y) \in \blockp{W_2  \mid G}$. By definition, $Y$ is obtained
from $X$ by adding a vertex from $W_2$. Since $a \notin V(Y)$, this violates the maximality of $X'$.
\qed
\end{proof}

\subsection{Partition model from superepisodes}\label{sec:super}

In the previous section we considered a model predicting the support of an
episode based on two smaller episodes. In this section we approach the ranking
from another perspective. Namely, we try to predict the support of $G$ using superepisodes
of $G$.

In order to motivate this consider the following example.

\begin{example}
Consider two episodes $G_1$ and $G_2$ given in Figure~\ref{fig:toysuper}.
Episode $G_2$ is a superepisode of $G_1$.
Assume that in our dataset, event $b$ occurs often once $a$ has occurred.  This
is to say that if we are in $H_2$ in either $\mach{G_1}$ or $\mach{G_2}$ we are
likely to move soon to $H_4$.

Assume also that occurrence of $a$ after $b$ follows the independence model,
or that $b$ occurs rarely without $a$ in front of it. In both cases the
elevated support of $G_1$ can be explained by the fact that $b$ follows often
after $a$. This means that we can explain the elevated support of $G_1$ if we know that $G_2$ 
has an elevated support, and consequently we should assign $G_1$ a low rank.
\end{example}

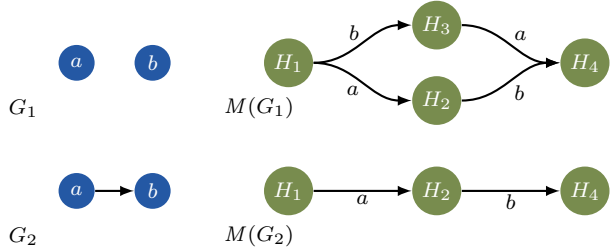
\begin{figure}[ht!]
\begin{center}
\begin{tikzpicture}[baseline]
\node[exnode] (a1) at (0, 0) {$a$};
\node[exnode] (a2) at (1, 0) {$b$};
\node at (-0.7, -0.6) {$G_1$};
\end{tikzpicture}\hspace{0.5cm}
\begin{tikzpicture}[auto, baseline]
\node[machnode] at (7.8*0.25, 0) (n2) {$H_1$};
\node[machnode] at (7.8*0.5, -0.5) (n3) {$H_2$};
\node[machnode] at (7.8*0.5, 0.5) (n4) {$H_3$};
\node[machnode] at (7.8*0.75, 0) (n5) {$H_4$};

\draw[machedge, in = 180, out = 0] (n2) edge node[swap, inner sep = 1pt] {$a$} (n3);
\draw[machedge, in = 180, out = 0] (n2) edge node[inner sep = 1pt] {$b$} (n4);
\draw[machedge, in = 180, out = 0] (n3) edge node[swap, inner sep = 1pt] {$b$} (n5);
\draw[machedge, in = 180, out = 0] (n4) edge node[inner sep = 1pt] {$a$} (n5);
\node at (7.8*0.2, -0.6) {$\mach{G_1}$};
\end{tikzpicture}\\[0.5cm]

\begin{tikzpicture}[baseline]
\node[exnode] (a1) at (0, 0) {$a$};
\node[exnode] (a2) at (1, 0) {$b$};
\draw[exedge] (a1) -- (a2);
\node at (-0.7, -0.6) {$G_2$};
\end{tikzpicture}\hspace{0.5cm}
\begin{tikzpicture}[auto, baseline]
\node[machnode] at (7.8*0.25, 0) (n2) {$H_1$};
\node[machnode] at (7.8*0.5, -0) (n3) {$H_2$};
\node[machnode] at (7.8*0.75, 0) (n5) {$H_4$};

\draw[machedge, in = 180, out = 0] (n2) edge node[swap, inner sep = 1pt] {$a$} (n3);
\draw[machedge, in = 180, out = 0] (n3) edge node[swap, inner sep = 1pt] {$b$} (n5);
\node at (7.8*0.2, -0.6) {$\mach{G_2}$};
\end{tikzpicture}
\end{center}
\caption{Episode $G_1$ and $G_2$ and the corresponding machines $\mach{G_1}$, $\mach{G_2}$.}
\label{fig:toysuper}
\end{figure}

Assume two episodes $G = (V, E_1, \lab{})$ and $H = (V, E_2,
\lab{})$ such that $E_1 \subsetneq E_2$. If $(W, E_2(W))$ is a prefix graph of $H$, then
$(W, E_1(W))$ is also a prefix graph of $G$. This allows us to define a mapping $\rho$ 
from $\pre{H}$ to $\pre{G}$ by setting $\rho((W, E_2(W))) = (W, E_1(W))$.
Moreover, if $x$ is a sink in $(W, E_2(W))$, then it is also a sink in $(W, E_1(W))$.
This immediately implies $\rho$ can be viewed as a graph homomorphism from $\mach{H}$
to $\mach{G}$, essentially making $\mach{H}$ a subgraph of $\mach{G}$.
We can now define the set of edges of
$\mach{G}$ for our partition model to be the edges in $\mach{H}$.  More
formally,

\begin{definition}
Given two episodes $G = (V, E_1, \lab{})$ and $H = (V, E_2, \lab{})$
such that $E_1 \subsetneq E_2$,
define a subset of edges $\blocks{H \mid G}$ of a machine $M = \mach{G}$,
\[
	\blocks{H \mid G} = \set{\rho(X, Y) \in E(M) \mid X \neq \emptyset, (X, Y) \in E(\mach{H})},
\]
that is, $\blocks{H \mid G}$ contains the edges from non-source vertices that can be also found in $\mach{H}$.
\end{definition}

We can now define $C_1 = \blocks{H \mid G}$ to be the first set of edges and
$C_2 = \emptyset$. Note that, similarly to the prefix graph approach from the
previous section, $C_1$ will not contain any edges from the source state.  The
rationale here is the same: transitions from the source state indicate
beginning of an episode while we are interested in modelling how fast we can
find the complete episode once we have found the first label. Also note that
since we require that $E_1 \neq E_2$, we will have at least one edge $(H, F)
\in E(\mach{G})$ such that $H \neq \emptyset$ and $(H, F)$ is not contained $C_1$.

\begin{example}
Consider an episode $G$ given in Figure~\ref{fig:toyprefix} 
and also the corresponding $\mach{G}$ given in Figure~\ref{fig:toymachine}.

Assume a serial episode $H = a \to b \to c \to d$. Then
\[
	\blocks{H \mid G} = \set{(H_2, H_3),\, (H_3, H_5),\, (H_5, H_6)},
\]
where $H_i$ are given in  Figure~\ref{fig:toymachine}.

On the other hand, if we set $H = a \to c \to b \to d$. Then
\[
	\blocks{H \mid G} = \set{(H_2, H_4),\, (H_4, H_5),\, (H_5, H_6)}\quad.
\]
\end{example}

Let $H$ be a superepisode of $G$ and let $C_1$ be the edges 
as constructed above. Write $M(C_1, \emptyset)$ to be the partition model.
We define the rank to be
\[
    \rpart{G; H} = \rank{G \mid M(C_1, \emptyset)}\quad.
\]
This rank can be computed using Proposition~\ref{prop:partcascade}.
In our experiments, we use the smallest rank induced by a superepisode in our
candidate set.

\subsection{Combining ranks}\label{sec:partrank}

Now that we have defined several different partition models, we propose a
simple approach to combine these models into a single rank.

To that end, assume that we have a collection $\mathcal{C}$ of episodes that we wish to rank.
These candidate episodes are obtained, for example, by mining frequent closed episodes.
For a given episode $G \in \mathcal{C}$,
let $\mathcal{P} = \pre{G} \setminus \set{G(\emptyset), G}$ be the prefix graphs
without the empty or the full prefix graph. Also, let $\mathcal{Q}$ be the proper superepisodes
of $G$ in $\mathcal{C}$ having the same vertices as $G$. We then compute the rank
by taking the smallest rank among all partition models,
\[
    \rpart{G} = \min(\min_{H \in \mathcal{P}} \rpart{G ; H},  \min_{H \in \mathcal{Q}} \rpart{G ; H})\quad.
\]
That is, if we can explain the support of $G$ by either a single prefix model
or a single superepisode in $Q$, then we will deem $G$ as redundant.

This approach mimics the approach of~\citet{webb:10:self-sufficient}, where
itemsets are filtered by comparing the observed support against the best
2-partition model.

\paragraph{Computational complexity} Finally, let us conclude this section with
a short discussion about computational complexity. Assume that we have an episode $G$
with $n$ nodes. Let $m = \abs{E(\mach{G})}$ be the number of edges in $\mach{G}$.

Using the partition model is
a two-step process, the first step is to find the parameters while the second
step is to compute the rank.  The first step uses iterative gradient descent,
for example, Newton-Raphson descent that requires $O(n^{2.373})$ time
for Hessian inversion and $O(n^2 + m)$ time for
constructing the matrix and gradient. The dominating term will depend on structure of the episode.
For example, for serial episodes we have $m = n$.
For general episodes we must have $m \leq 2^n$ and 
for parallel episodes we have $m = 2^n$.

In order to compute $\rpart{G}$ we need to loop over all prefix episodes. Again, the
number of such episodes depends on $G$. For serial episodes there are only $n + 1$ such episodes,
whereas a parallel episode has $2^n$ prefix episodes. The parallel episode case is the worst case
since there are only $2^n$ subepisodes in any $G$.

This implies that in theory computing this rank may not scale for large
episodes, especially if they are parallel. Fortunately, in practice, most
episodes are small and for these cases our approach remains feasible.

\section{Related Work}\label{sec:related}

\emph{Discovering episodes:} Episode discovery was introduced
by~\citet{mannila:97:discovery} where the authors consider episodes defined as DAGs
and consider two concepts of support: the first one based on sliding windows of fixed length
and the second one based on minimal windows. Unfortunately, the number of minimal windows is not
monotonic in general---however this can be fixed by considering the maximal
number of non-overlapping windows, see for example~\citep{laxman:07:fast}.
Mining general episodes can be intricate and computationally heavy, for example, discovering whether a
sequence covers a general episode is \np-hard~\cite{tatti:11:mining}.
Consequently, research focus has been into mining subclasses of episodes, such
as, episodes with unique labels~\cite{achar:12:discovering,pei:06:discovering},
and strict episodes~\cite{tatti:12:mining}.  A miner for general episodes that
can handle simultaneous events was proposed by~\citet{tatti:11:mining}.  An
important subclass of episodes are serial episodes or sequential patterns.  A
widely used miner for mining closed serial episodes was suggested
by~\citet{wang:04:bide}.

\emph{Ranking episodes:} Unlike with itemset mining, ranking episodes based on
surprisingness is underdeveloped. The most straightforward way of ranking
episodes, reviewed in Section~\ref{sec:independence}, by comparing the support
against the independence model, was introduced by~\citet{gwadera:05:reliable}.
Using Markov models instead of the independence model to rank serial episodes
was suggested by~\citet{gwadera:05:markov}. Both of these pioneer works focus
on ranking episodes by analyzing support based on a sliding window, that is, the
input dataset is a single sequence and the support of an episode is the number
of sliding windows of fixed length that cover the episode. Interestingly
enough, this scenario generates technical complications since the windows are
no longer independent, unlike in the setup where we have many sequences and we
assume that they are generated independently. These complications can be
overcome but they require additional computational steps.  Instead of using
windows of fixed length, ranking based on minimal window lengths with respect
to the independence model was suggested by~\citet{tatti:14:mining}.  Ranking
serial episodes allowing multiple labels using the independence model was
suggested by~\citet{DBLP:conf/icdm/Low-KamRKP13}.  \citet{achar:12:discovering}
also considered a measure that downranks the episode if there is
a non-edge $(x, y)$ that occurs rarely, which suggests that we should augment
the episode with the edge $(y, x)$. 

In related work, \citet{mannila:00:global} consider general episodes as
generative models for sequences. They generate short sequences by selecting a
subset of events from an episode and ordering events with a random order
compatible with the episode. They do not allow gaps and only one pattern is
responsible for generating a single sequence.

Finally, SQS and GoKrimp, pattern set mining approaches for discovering serial episodes were
respectively introduced by~\citet{tatti:12:long} and by~\citet{DBLP:journals/sadm/LamMFC14}.
The idea behind the approach is to find a
small set of serial episodes that model the data well. In order to do that the
authors constructed a model given a set of episodes and used a posteriori
probability of the model to score the episode set. The authors then used a
heuristic search to find a set with good episodes. In general, the goal of our
approach and the is the same: reducing
the redundancy in patterns. From a technical point of view, the approaches are
different: in this work we rank episodes based on how surprising their support
is while the pattern set mining methods select episodes based on
how well we can model the data using the episodes. Moreover, we work with
general strict episodes while the current pattern set approaches limit
themselves to serial episodes. Extending these miners to general episodes is an
interesting future line of research. However, it is highly non-trivial due to
the fact that the score, the algorithm for computing the score, and the
mining algorithm are specifically designed for serial episodes.

\section{Experiments}\label{sec:exp}

\emph{Datasets:}
In our experiments we used 3 synthetic datasets and 3 text datasets.
The sizes of the datasets are given in Table~\ref{tab:basic}.

The first synthetic dataset, \emph{Plant}, was created as follows. We generated
$10\,000$ sequences of length randomly selected from a uniform distribution
between $20$ and $30$.  A single event in each sequence was generated from a
uniform distribution of $990$ events.  We planted two serial episodes and one
general in the data. The first episode, a serial episode of 4 vertices was
planted with no gaps into a randomly selected sequence $200$ times. The second
episode, a serial episode of 2 vertices was planted with no gaps into a
randomly selected sequence $20$ times.  The third episode, given in
Table~\ref{tab:plant}, was planted $10$ times with no gaps, the order of events
$n$ and $m$ was picked uniformly.  We made sure that the events used in planted
patterns did not occur in the noise.  This gave us an alphabet of size $1000$.

The second synthetic dataset, \emph{Plant2}, was created as follows. We
generated $10\,000$ sequences of length randomly selected from a uniform
distribution between $20$ and $30$.  A single event in each sequence was
generated from a uniform distribution of $1000$ events.  We planted two serial
episodes with 3 vertices with no gaps $400$ times.

The third synthetic dataset, \emph{Gap}, was created as follows. Similarly to \emph{Plant}, we generated
$10\,000$ sequences of length between $20$ and $30$. An event in each sequence was generated from a
uniform distribution of $996$ events. We planted one serial episode of 4
events into the data $200$ times.  We set the probability of the next event
being a noise event to be $p$, this made the average gap length to be $p/(1 -
p)$. We varied $p$ from $0$ to $0.8$ with $0.05$ increments. We did not plant
events if they did not fit into a sequence.

Our fourth dataset, \emph{Moby}, is the novel Moby Dick by Herman
Melville.\!\footnote{\url{http://www.gutenberg.org/etext/15}.}  Our fifth
dataset, \textit{JMLR} consists of abstracts of papers from the Journal of
Machine Learning Research website,\!\footnote{\url{http://jmlr.csail.mit.edu/}}
Our final dataset, \emph{Addresses}, consists of inaugural addresses of the
presidents of the United States.\!\footnote{\url{http://www.bartleby.com/124/}}
We processed the datasets by stemming the words and removing the stop words.
We further split the text into sequences such that a sequence corresponds to a
single sentence.

\begin{table}[ht!]
\begin{tabular*}{\columnwidth}{@{\extracolsep{\fill}}lrrrrr}
\toprule
Dataset & $\abs{\mathcal{S}}$ & $\abs{\text{events}}$ & $\sigma$ & $\abs{\mathcal{C}}$ & time\\
\midrule
\emph{Plant}     & 10\,000 & 249\,955 & 10 & 43\,029 & 56s \\
\emph{Plant2}    & 10\,000 & 249\,736 & 10 & 46\,329 & 50s \\
\emph{Gap}       & 10\,000 & 250\,150 & -- &      -- &  -- \\[1mm]
\emph{Addresses} &    5584 &  62\,066 &  5 & 19\,367 & 12s \\
\emph{JMLR}      &    5986 &  75\,646 &  5 & 49\,951 & 46s \\
\emph{Moby}      & 13\,987 & 105\,671 &  5 & 17\,550 & 26s \\
\bottomrule
\end{tabular*}
\caption{Basic characteristics of datasets, frequency thresholds, the numbers of discovered
episodes, and running time needed to rank the episodes. The number of events for \emph{Gap} is an average over 17 datasets.}
\label{tab:basic}
\end{table}

\emph{Setup:} We mined closed strict episodes from each dataset, except \emph{Gap}, with a miner
given by~\citet{tatti:12:mining}. As frequency thresholds we used 5 for text datasets and
10 for the synthetic dataset. The amount of discovered patterns, $\abs{\mathcal{C}}$, is given in
Table~\ref{tab:basic}. We then proceeded by ranking each episode first by
independence model and then by the partition model.\!\footnote{The implementation is available at \url{http://research.ics.aalto.fi/dmg/}.}

\emph{Results:}
Our main goal is to compare $\rpart{G}$, ranks
given by the partition model, against the baseline ranks given by the independence model,
$\rind{G}$. 

Let us first consider the synthetic dataset \emph{Gap}. We considered ranks for
3 different episodes, given in Figure~\ref{fig:gap}, the planted serial episode
$G_1$, the planted episode with additional noise event $G_2$, here we took an
average rank of 10 such episodes, and finally $G_3$ a non-trivial subepisode of
$G_1$.  The ranks $\rind{G_1}$ and $\rind{G_3}$ were outside floating point
range.  The remaining ranks are given in Figure~\ref{fig:gap} as a function of
the gap probability.  Let us first consider $G_2$ and $G_3$. Unlike the
independence model, the partition model predicts the support accurately for
these patterns which results in a low rank. Episode $G_2$ is predicted
accurately due to a partition of $G_2$ to $G_1$ and the noise label while $G_3$
is predicted accurately due to $G_1$ being a superepisode of $G_3$.  As
expected, the rank $\rpart{G_1}$ remains high as there are no partition model
that can explain this pattern. This rank goes down as the average gap length increases as the
planted pattern becomes more and more explainable by the independence model.

\begin{figure}
\begin{tikzpicture}
\node[smallnode] (n1) at (0, 0) {$a$};
\node[smallnode, base right = 0.4cm of n1] (n2) {$b$};
\node[smallnode, base right = 0.4cm of n2] (n3) {$c$};
\node[smallnode, base right = 0.4cm of n3] (n4) {$d$};
\node[smallnode, base left = 0.0cm of n1] (n0) {$G_1$};

\draw (n1.mid east) edge[smalledge, out = 0, in = 180] (n2.mid west);
\draw (n2.mid east) edge[smalledge, out = 0, in = 180] (n3.mid west);
\draw (n3.mid east) edge[smalledge, out = 0, in = 180] (n4.mid west);

\node[smallnode] (m1) at (0, -1) {$a$};
\node[smallnode, base right = 0.4cm of m1] (m2) {$b$};
\node[smallnode, base right = 0.4cm of m2] (m3) {$c$};
\node[smallnode, base right = 0.4cm of m3] (m4) {$d$};
\node[smallnode, base right = 0.4cm of m4] (m5) {$x$};
\node[smallnode, base left = 0.0cm of m1] (m0) {$G_2$};

\draw (m1.mid east) edge[smalledge, out = 0, in = 180] (m2.mid west);
\draw (m2.mid east) edge[smalledge, out = 0, in = 180] (m3.mid west);
\draw (m3.mid east) edge[smalledge, out = 0, in = 180] (m4.mid west);
\draw (m4.mid east) edge[smalledge, out = 0, in = 180] (m5.mid west);

\node[smallnode] (k1) at (0, -2) {$a$};
\node[smallnode] at (1, -1.6) (k2) {$b$};
\node[smallnode] at (1, -2.4) (k3) {$c$};
\node[smallnode, base right = 1.8cm of k1] (k4) {$d$};
\node[smallnode, base left = 0.0cm of k1] (k0) {$G_3$};

\draw (k1.mid east) edge[smalledge, out = 0, in = 180] (k2.mid west);
\draw (k1.mid east) edge[smalledge, out = 0, in = 180] (k3.mid west);
\draw (k3.mid east) edge[smalledge, out = 0, in = 180] (k4.mid west);
\draw (k2.mid east) edge[smalledge, out = 0, in = 180] (k4.mid west);

\end{tikzpicture}
\begin{tikzpicture}[baseline]
\begin{axis}[xlabel={gap probability}, ylabel= {score},
    height = 2.8cm,
    width = 2.8cm,
    scale only axis,
    cycle list name=yaf,
	every axis plot post/.style = {},
	ymin = 0,
	ymax = 60,
	every axis legend/.append style = {at = {(-0.02, 0.1)}, anchor = south west}
    ]
\addplot+[mark size = 1, fill opacity = 0.2, mark = *]
	table[x expr = {\thisrowno{0}}, y expr = {\thisrowno{1}}, header = false]  {gap.dat};
\addplot+[fill opacity = 0.2, mark = x]
	table[x expr = {\thisrowno{0}}, y expr = {\thisrowno{3}}, header = false]  {gap.dat};
\addplot+[fill opacity = 0.2, mark = +]
	table[x expr = {\thisrowno{0}}, y expr = {\thisrowno{4}}, header = false]  {gap.dat};

\legend{$\rpart{G_2}$, $\rpart{G_3}$, $\rind{G_2}$}

\pgfplotsextra{\yafdrawaxis{0}{0.8}{0}{60}}
\end{axis}
\end{tikzpicture}
\begin{tikzpicture}[baseline]
\begin{axis}[xlabel={gap probability}, ylabel= {$\rpart{G}$},
    height = 2.8cm,
    width = 2.8cm,
    scale only axis,
    cycle list name=yaf,
	every axis plot post/.style = {},
	ymin = 0,
	every axis legend/.append style = {at = {(0, 0.1)}, anchor = south west}
    ]
\addplot+[mark size = 1, fill opacity = 0.2, mark = *, mark options = {line width = 0.5pt}]
	table[x expr = {\thisrowno{0}}, y expr = {\thisrowno{2}}, header = false]  {gap.dat};

\legend{$\rpart{G_1}$}

\pgfplotsextra{\yafdrawaxis{0}{0.8}{0}{660}}
\end{axis}
\end{tikzpicture}
\caption{Episodes and their ranks in \emph{Gap} datasets. The plots for $G_2$ represent an average of 10 different episodes, each of them having different noise event $x$.
The scales of $y$-axis of plots are different.  The values of $\rind{G_1}$ and $\rind{G_3}$ were outside floating point range.}
\label{fig:gap}
\end{figure}
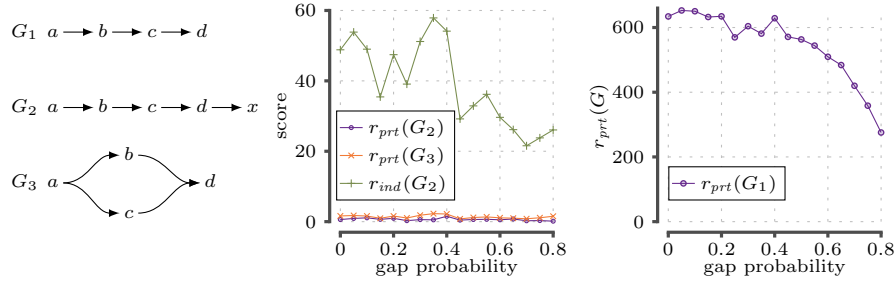
Let us now look at the top episodes in \emph{Plant} dataset, given in Table~\ref{tab:plant}. 
The top episode having the largest
$\rind{G}$ is the planted serial episode of 4 vertices $a \to b \to c \to d$.
The second episode is the planted general episode.
However, the next 5 episodes are of form $a \to b \to c \to d \to x$, where $x$
is a noise label. These episodes have abnormally high support because of the
original high support of the planted pattern. The 8th episode according to
$\rind{G}$ is the second planted episode, namely $e \to f$.  Let us now look at
the top episodes according to $\rpart{G}$. The top 3 episodes
are the planted episodes. The remaining episodes are either parallel episodes
or serial episodes containing 2 events or episodes
of form $a \to b \to c \to d \to x$, where $x$ is a noise label.
There is a clear difference between the score values. While the rank for the first
three episodes was $78$--$10^{308}$, the ranks for the
remaining episodes varied between $0$ and $14$. 
In other words, $\rpart{}$ successfully downgraded the freerider episodes
that had significant $\rind{}$. Some of the freerider episodes still have a significantly
large rank. This is due to the multiple hypothesis phenomenon: if we test large amount of patterns,
then some of them will have abnormal support just by chance.

We observe similar behaviour in \emph{Plant2} dataset. The top-8 episodes in
\emph{Plant2} according to $\rpart{}$ are the 2 planted serial episodes (ranked
as 2nd and 4th) and the 6 serial subepisodes with 2 vertices. The ranks of these episodes
were $1205$--$3704$. The remaining episodes
were ranked between $0$--$15$. On the other hand, $\rind{G}$ ranked the 2 planted patterns as top-2 episodes.
The next 1089 episodes contained either vertices from both patterns,
or several vertices from one pattern and one noise event. These episodes had ranks $19$--$3704$.
Episodes $G_1 = a \to b \to c \to d \to e \to f$ and $G_2 = (a \to b \to c), (d \to e \to f)$
had ranks $\rind{G_1} = 212$ (17th) and $\rind{G_2} = 289$ (14th) whereas the partition model gave the ranks $\rpart{G_1} = 4.8$ (5575th) and $\rpart{G_2} = 0.0005$ (45284th).
The remaining episodes
were ranked between $0$--$15$.

\begin{table}
\caption{Top episodes in \emph{Plant} dataset. The symbols $x$ and $y$ represent
noise events. The rank for the first episode with respect to the independence
model is outside the floating point range. }
\label{tab:plant}
\begin{tabular*}{\columnwidth}{@{\extracolsep{\fill}}lrrrlrr}
\toprule
\multicolumn{3}{l}{Independence model} &&
\multicolumn{3}{l}{Partition model} \\
\cmidrule{1-3} \cmidrule{5-7}
Rank & Episode type & $\rind{G}$ &&
Rank & Episode type & $\rpart{G}$ \\
\midrule
1.  & $a \epito b \epito c \epito d$ & $\infty$ &&
1.  & $a \epito b \epito c \epito d$ & $10^{308}$ \\

2. & \begin{tikzpicture}[baseline]
\node[smallnode] (k1) at (0, 0) {$k$};
\node[smallnode] at (0.8, -0.15) (k2) {$m$};
\node[smallnode] at (0.8, 0.15) (k3) {$n$};
\node[smallnode, base right = 1.4cm of k1] (k4) {$l$};

\draw (k1.mid east) edge[smalledge, out = 0, in = 180] (k2.mid west);
\draw (k1.mid east) edge[smalledge, out = 0, in = 180] (k3.mid west);
\draw (k3.mid east) edge[smalledge, out = 0, in = 180] (k4.mid west);
\draw (k2.mid east) edge[smalledge, out = 0, in = 180] (k4.mid west);
\end{tikzpicture} & 249 &&
2.  & $e \epito f$ & $128$ \\

3.--7.  & $a \epito b \epito c \epito d \epito x$ & $184$--$185$ &&

3. &
\begin{tikzpicture}[baseline]
\node[smallnode] (k1) at (0, 0) {$k$};
\node[smallnode] at (0.8, -0.15) (k2) {$m$};
\node[smallnode] at (0.8, 0.15) (k3) {$n$};
\node[smallnode, base right = 1.4cm of k1] (k4) {$l$};

\draw (k1.mid east) edge[smalledge, out = 0, in = 180] (k2.mid west);
\draw (k1.mid east) edge[smalledge, out = 0, in = 180] (k3.mid west);
\draw (k3.mid east) edge[smalledge, out = 0, in = 180] (k4.mid west);
\draw (k2.mid east) edge[smalledge, out = 0, in = 180] (k4.mid west);
\end{tikzpicture} & 78\\

8. & $e \epito f$ & $128$ &&
4.--  & $a \epito b \epito c \epito d \epito x$  & $0$--$14$ \\

9.-- &  $x \epito y$\quad or\quad $x, y$ & $2$--$14$ 

&&
& or\quad $x \epito y$\quad or\quad $x, y$ \\

\bottomrule
\end{tabular*}
\end{table}

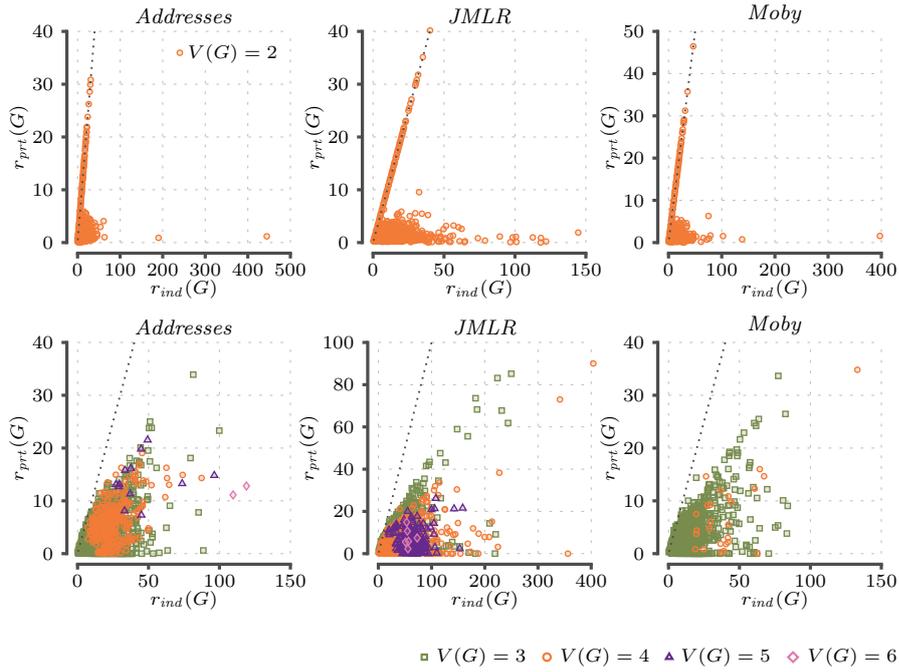
\begin{figure}[ht!]
\newlength{\imgsize}
\setlength{\imgsize}{2.8cm}
\begin{tikzpicture}[baseline]
\begin{axis}[xlabel={$\rind{G}$}, ylabel= {$\rpart{G}$},
	title = {\emph{Addresses}},
    height = \imgsize,
    width = \imgsize,
    scale only axis,
    cycle list name=yaf,
	every axis plot post/.style = {},
    xmax = 500, xmin = 0, ymax = 40, ymin = 0,
	xtick = {0, 100,..., 500},
	legend entries = {$V(G) = 2$},
	legend style = {draw = none},
	clip mode = individual
    ]
\addplot[only marks, mark size = 1, fill opacity = 0.2, mark = *, yafcolor2, mark options = {line width = 0.5pt}]
	table[x expr = {-\thisrowno{0}}, y expr = {-\thisrowno{1}}, header = false]  {addresspairsranksparse.dat};
\addplot[yafaxiscolor, line width = 0.7pt, dotted] coordinates {(40, 40) (0, 0)};
\pgfplotsextra{\yafdrawaxis{0}{500}{0}{40}}
\end{axis}
\end{tikzpicture}\hspace{-5mm}\hfill%
\begin{tikzpicture}[baseline]
\begin{axis}[xlabel={$\rind{G}$}, ylabel= {$\rpart{G}$},
	title = {\emph{JMLR}},
    width = \imgsize,
    height = \imgsize,
    scale only axis,
    cycle list name=yaf,
	every axis plot post/.style = {},
    xmax = 150, xmin = 0, ymax = 40, ymin = 0,
	clip mode = individual
    ]
\addplot[only marks, mark size = 1, fill opacity = 0.2, mark = *, yafcolor2, mark options = {line width = 0.5pt}]
	table[x expr = {-\thisrowno{0}}, y expr = -\thisrowno{1}, header = false]  {jmlrpairsranksparse.dat};
\addplot[yafaxiscolor, line width = 0.7pt, dotted] coordinates {(40, 40) (0, 0)};
\pgfplotsextra{\yafdrawaxis{0}{150}{0}{40}}
\end{axis}
\end{tikzpicture}\hspace{-5mm}\hfill%
\begin{tikzpicture}[baseline]
\begin{axis}[xlabel={$\rind{G}$}, ylabel= {$\rpart{G}$},
	title = {\emph{Moby}},
    width = \imgsize,
    height = \imgsize,
    scale only axis,
    cycle list name=yaf,
	every axis plot post/.style = {},
    xmax = 400, xmin = 0, ymax = 50, ymin = 0,
	ytick = {0, 10,..., 50},
	clip mode = individual
    ]
\addplot[only marks, mark size = 1, fill opacity = 0.2, mark = *, yafcolor2, mark options = {line width = 0.5pt}]
	table[x expr = {-\thisrowno{0}}, y expr = -\thisrowno{1}, header = false]  {mobypairsranksparse.dat};
\addplot[yafaxiscolor, line width = 0.7pt, dotted] coordinates {(50, 50) (0, 0)};
\pgfplotsextra{\yafdrawaxis{0}{400}{0}{50}}
\end{axis}
\end{tikzpicture}

\begin{tikzpicture}[baseline]
\begin{axis}[xlabel={$\rind{G}$}, ylabel= {$\rpart{G}$},
	title = {\emph{Addresses}},
    height = \imgsize,
    width = \imgsize,
    scale only axis,
    cycle list name=yaf,
	every axis plot post/.style = {},
    scatter/classes = {3={yafcolor3, mark=square*},4={yafcolor2, mark=*},5={yafcolor1, mark=triangle*, mark size=1.5},6={yafcolor8, mark=diamond*, mark size=1.5}},
    xmax = 150, xmin = 0, ymax = 40, ymin = 0,
	legend to name = {lgd:scatter},
	legend style = {draw = none},
	legend columns = 4,
	legend entries = {$V(G) = 3$, $V(G) = 4$, $V(G) = 5$, $V(G) = 6$},
	clip mode = individual
    ]
\addplot[scatter, scatter src = explicit symbolic, only marks, mark size = 1, fill opacity = 0.2, mark options = {line width = 0.5pt}]
	table[x expr = {-\thisrowno{0}}, y expr = -\thisrowno{1}, meta index = 2, header = false] {addressranksparse.dat};
\addplot[yafaxiscolor, line width = 0.7pt, dotted, forget plot] coordinates {(40, 40) (0, 0)};
\pgfplotsextra{\yafdrawaxis{0}{150}{0}{40}}
\end{axis}
\end{tikzpicture}\hspace{-5mm}\hfill%
\begin{tikzpicture}[baseline]
\begin{axis}[xlabel={$\rind{G}$}, ylabel= {$\rpart{G}$},
	title = {\emph{JMLR}},
    height = \imgsize,
    width = \imgsize,
    cycle list name=yaf,
    scale only axis,
	every axis plot post/.style = {},
    scatter/classes = {6={yafcolor8, mark=diamond*, mark size=1.5},2={yafcolor2},3={yafcolor3, mark=square*},4={yafcolor2, mark=*},5={yafcolor1, mark=triangle*, mark size=1.5}},
    xmax = 400, xmin = 0, ymax = 100, ymin = 0,
	clip mode = individual
    ]
\addplot[scatter, scatter src = explicit symbolic, only marks, mark size = 1, fill opacity = 0.2, mark = *, mark options = {line width = 0.5pt}]
	table[x expr = {-\thisrowno{0}}, y expr = -\thisrowno{1}, meta index = 2, header = false]  {jmlrranksparse.dat};
\addplot[yafaxiscolor, line width = 0.7pt, dotted] coordinates {(100, 100) (0, 0)};
\pgfplotsextra{\yafdrawaxis{0}{400}{0}{100}}
\end{axis}
\end{tikzpicture}\hspace{-5mm}\hfill%
\begin{tikzpicture}[baseline]
\begin{axis}[xlabel={$\rind{G}$}, ylabel= {$\rpart{G}$},
	title = {\emph{Moby}},
    width = \imgsize,
    height = \imgsize,
    scale only axis,
    cycle list name=yaf,
	every axis plot post/.style = {},
    scatter/classes = {6={yafcolor4, mark=diamond*, mark size=1.5},2={yafcolor2},3={yafcolor3, mark=square*},4={yafcolor2, mark=*},5={yafcolor1, mark=triangle*, mark size=1.5}},
    xmax = 150, xmin = 0, ymax = 40, ymin = 0,
	clip mode = individual
    ]
\addplot[scatter, scatter src = explicit symbolic, only marks, mark size = 1, fill opacity = 0.2, mark = *, mark options = {line width = 0.5pt}]
	table[x expr = {-\thisrowno{0}}, y expr = -\thisrowno{1}, meta index = 2, header = false] {mobyranksparse.dat};
\addplot[yafaxiscolor, line width = 0.7pt, dotted] coordinates {(40, 40) (0, 0)};
\pgfplotsextra{\yafdrawaxis{0}{150}{0}{40}}
\end{axis}
\end{tikzpicture}

\begin{flushright}
\begin{tikzpicture}
\node[draw, rectangle, yafcolor3, thick, inner sep = 1pt] (l1) {};
\node[right = 0pt of l1, font=\scriptsize] {$V(G) = 3$};

\node[draw, circle, yafcolor2, thick, inner sep = 1pt, right = 1.5cm of l1] (l2) {};
\node[right = 0pt of l2, font=\scriptsize] {$V(G) = 4$};

\node[draw, regular polygon, regular polygon sides = 3, yafcolor1, thick, inner sep = 0.5pt, right = 1.5cm of l2] (l3) {};
\node[right = 0pt of l3, font=\scriptsize] {$V(G) = 5$};

\node[draw, diamond, yafcolor8, thick, inner sep = 1pt, right = 1.5cm of l3] (l4) {};
\node[right = 0pt of l4, font=\scriptsize] {$V(G) = 6$};

\end{tikzpicture}
\end{flushright}
\caption{Partition model ranks $\rpart{G}$ as a function of $\rind{G}$ for text
datasets. The top row contains parallel episodes with 2 vertices. The bottom row contains
episodes with more than two vertices. The ranges of axis vary from figure to figure.}
\label{fig:scatter}
\end{figure}

\begin{table}
\caption{Kendall-$\tau$ coefficients of episodes ranked by $\rpart{}$ and $\rind{}$.}
\label{tab:kendall}
\begin{tabular*}{\columnwidth}{@{\extracolsep{\fill}}lrrr}
\toprule
Dataset & All & parallel, $V(G) = 2$ & $V(G) > 3$ \\
\midrule
\emph{Addresses} & 0.61 & 0.60 & 0.42 \\
\emph{JMLR} & 0.54 & 0.62 & 0.45 \\
\emph{Moby} & 0.66 & 0.59 & 0.38 \\
\bottomrule
\end{tabular*}
\end{table}

SQS miner~\citep{tatti:12:long} discovered the planted serial episode in all
\emph{Gap} datasets. In \emph{Plant} SQS discovered the two planted serial
episodes, but not the general planted episode, since SQS discovers only serial
episodes. Instead, SQS found the two serial superepisodes $k \to n \to m \to l$
and $k \to m \to n \to l$.

Let us now consider episodes discovered from text datasets. In
Figure~\ref{fig:scatter} we plot $\rpart{G}$ as a function of $\rind{G}$.  We
highlight parallel episodes with $2$ vertices by plotting them separately in
the top row while the bottom row contains the episodes with more than two
vertices. Note that we omitted serial episodes of size $2$ since both ranks
will produce an equal score, $\rind{G} = \rpart{G}$, since there are no proper
superepisodes for an episode $G$ and the only prefix partition is actually
equal to the independence model.

The results demostrate that $\rpart{G}$ is typically much smaller than
$\rind{G}$. This implies that there are lot of patterns whose abnormally high support can be
justified by a partition model.
In the top
row of Figure~\ref{fig:scatter} we see that the parallel episodes of size 2 are typically
considered redundant by $\rpart{G}$ because typically the serial counterpart
of the episode can explain well the behaviour of the parallel episode.  For
certain parallel episodes, the rank remains the same by design as there are no
serial counterpart episodes in the mined collection.

The Kendall-$\tau$ coefficients given in Table~\ref{tab:kendall} imply that 
$\rpart{}$ and $\rind{}$ are correlated. The correlation is weaker for larger
episodes than for parallel episodes of size 2. This is because $\rpart{G} =
\rind{G}$ if $G$ is a parallel episode of size $2$ and does not have a frequent
serial episode.

The top episodes according to $\rpart{G}$, given in Table~\ref{tab:top10} in
the text datasets were short serial episodes of words that occur often
together.  This is an expected result as these episodes represent common
expressions. For comparison, the top-10 patterns obtained by SQS are given in Table~\ref{tab:topsqs}.
While serial episodes are favored by $\rpart{G}$, there are non-serial
episodes that have high rank, for example, $G_1 = $ \emph{east}, \emph{west} in \emph{Addresses}
has rank $\rpart{G_1} = 30$ (42nd), and $G_2 = $ (\emph{subgroup}$\to$\emph{discoveri}), \emph{rule}
in \emph{JMLR} has rank $\rpart{G_2} = 25$ (376th).

\tikzexternaldisable
\begin{table}
\caption{Top-10 episodes according to $\rpart{G}$ and $\rind{G}$.}
\label{tab:top10}
\setlength{\tabcolsep}{0pt}
\begin{tabular*}{\textwidth}{@{\extracolsep{\fill}}lllllll}
\toprule
\multicolumn{2}{l}{ranked by $\rind{G}$} & $\rind{}$ & $\rpart{}$ & ranked by $\rpart{G}$ & $\rind{}$ & $\rpart{}$ \\
\cmidrule(r{2pt}){1-4}\cmidrule{5-7}
\multicolumn{2}{l}{\emph{Addresses}}\\
\cmidrule(r{2pt}){1-4}\cmidrule{5-7}

1.&unit\epitoc state &
931 & 931 &
unit\epitoc state &
931 & 931
\\

2.&unit state &
445 & 1.2 &
fellow\epitoc citizen &
256 & 256
\\

3.&fellow\epitoc citizen &
256 & 256 &
constitut\epitoc state &
97 & 97
\\

4.&fellow citizen &
190 & 0.9 &
four\epitoc year  &
79 & 79
\\

5.&
\begin{tikzpicture}[baseline]
\node[mediumnode] (n0) {preserv};
\node[mediumnode, base right = 0.4cm of n0] (n1) {protect};
\node[mediumnode, base right = 0.4cm of n1] (n2) {defend};
\node[mediumnode] (n3) at (0.8, -0.4) {constitut};
\node[mediumnode, base right = 0.4cm of n3] (n4) {unit};
\node[mediumnode, base right = 0.4cm of n4] (n5) {state};

\draw (n0.mid east) edge[smalledge, out = 0, in = 180] (n1.mid west);
\draw (n1.mid east) edge[smalledge, out = 0, in = 180] (n2.mid west);
\draw (n3.mid east) edge[smalledge, out = 0, in = 180] (n4.mid west);
\draw (n4.mid east) edge[smalledge, out = 0, in = 180] (n5.mid west);
\draw (n2) edge[smalledge, out = 235, in = 10, looseness = 0.3] (n3);
\end{tikzpicture} &

119 & 13 &
men\epitoc women &
77 & 77
\\

6.&
\begin{tikzpicture}[baseline]
\node[mediumnode, inner sep = 0pt] (n0) {best};
\node[mediumnode, base right = 0.35cm of n0, inner sep = 0pt] (n1) {abil};
\node[mediumnode, base right = 0.35cm of n1, inner sep = 0pt] (n2) {preserv};
\node[mediumnode, base right = 0.35cm of n2, inner sep = 0pt] (n3) {protect};
\node[mediumnode, below = 0.0cm of n3, inner xsep = 0pt] (n4) {defend};
\node[mediumnode] (n5) at (0.8, -0.3) {constitut};

\draw (n0.mid east) edge[smalledge, out = 0, in = 180] (n1.mid west);
\draw (n1.mid east) edge[smalledge, out = 0, in = 180] (n2.mid west);
\draw (n2.mid east) edge[smalledge, out = 0, in = 180] (n3.mid west);
\draw (n3) edge[smalledge, out = 0, in = 0, looseness = 2] (n4);
\end{tikzpicture} &

110 & 11 &
year\epitoc ago &
75 & 75
\\

7.&constitut\epitoc unit\epitoc state &
100 & 23 &
armi\epitoc navi &
63 & 63
\\

8.&constitut\epitoc state &
97 & 97 &
north\epitoc south  &
53 & 53
\\

9.&
\begin{tikzpicture}[baseline]
\node[mediumnode] (n0) {preserv};
\node[mediumnode, base right = 0.4cm of n0] (n1) {constitut};
\node[mediumnode, base right = 0.4cm of n1] (n2) {state};
\node[mediumnode] (n3) at (0.8, -0.4) {protect};
\node[mediumnode, base right = 0.4cm of n3] (n4) {defend};

\draw (n0.mid east) edge[smalledge, out = 0, in = 180] (n1.mid west);
\draw (n1.mid east) edge[smalledge, out = 0, in = 180] (n2.mid west);
\draw (n3.mid east) edge[smalledge, out = 0, in = 180] (n4.mid west);
\draw (n0) edge[smalledge, out = 270, in = 180] (n3.mid west);
\draw (n4.mid east) edge[smalledge, out = 0, in = 270] (n2);
\end{tikzpicture}

&

96 & 15 &
within\epitoc limit  &
52 & 52
\\

10.&unit\epitoc state constitut &
89 & 0.6 &
chief\epitoc magistr &
51 & 51
\\

\cmidrule(r{2pt}){1-4}\cmidrule{5-7}
\multicolumn{2}{l}{\emph{JMLR}}\\
\cmidrule(r{2pt}){1-4}\cmidrule{5-7}
1.&support\epitoc vector\epitoc machin &
$\infty$ &
357 &
support\epitoc vector &
440 &
440

\\
2.&
support\epitoc vector &
440 &
440 &
support\epitoc vector\epitoc machin &
$\infty$ &
357

\\
3.&
support\epitoc vector\epitoc machin\epitoc svm &
404 &
90 &
support\epitoc machin &
324 &
324

\\
4.&
support\epitoc vector\epitoc machin svm &
356 &
$10^{-3}$ &
vector\epitoc machin &
306 &
306

\\
5.&
reproduc\epitoc kernel\epitoc hilbert\epitoc space &
341 &
73 &
data\epitoc set &
284 &
284

\\
6.&
support\epitoc machin &
325 &
325 &
real\epitoc world &
260 &
260

\\
7.&
vector\epitoc machin  &
306 &
306 &
real\epitoc data &
213 &
213

\\
8.&
data\epitoc set &
284 &
284 &
state\epitoc art &
191 &
191

\\
9.&
real\epitoc world &
260 &
260 &
machin\epitoc learn &
190 &
190

\\
10.&
support\epitoc vector\epitoc svm &
250 &
85 &
bayesian\epitoc network &
166 &
166

\\

\cmidrule(r{2pt}){1-4}\cmidrule{5-7}
\multicolumn{2}{l}{\emph{Moby}}\\
\cmidrule(r{2pt}){1-4}\cmidrule{5-7}

1.&sperm\epitoc whale &
874 &
874 &
sperm\epitoc whale &
874 &
874
\\

2.&sperm whale &
397 &
1.6 &
mobi\epitoc dick &
359 &
359
\\

3.&mobi\epitoc dick &
359 &
359 &
old\epitoc man &
224 &
224
\\

4.&old\epitoc man &
224 &
224 &
mast\epitoc head &
186 &
186
\\

5.&mast\epitoc head &
187 &
187 &
white\epitoc whale &
179 &
179
\\

6.&white\epitoc whale &
179 &
179 &
right\epitoc whale &
131 &
131
\\

7.&head mast &
138 &
0.8 &
quarter\epitoc deck  &
96 &
96
\\

8.&seven\epitoc hundr\epitoc seventi\epitoc seventh &
133 &
35 &
captain\epitoc peleg &
86 &
86
\\

9.&right\epitoc whale &
131 &
131 &
chief\epitoc mate &
85 &
85
\\
10.&old man &
102 &
1.5 &
new\epitoc bedford &
82 &
82
\\

\bottomrule
\end{tabular*}
\end{table}

\tikzexternaldisable
\begin{table}
\caption{Top-10 episodes according to SQS.
The pattern $G$ in \emph{Moby} was a long episode,
such \epito funni \epito sporti \epito gami \epito jesti \epito joki \epito hoki \epito poki \epito lad \epito ocean, a litany repeated 3 times in the novel.}
\label{tab:topsqs}
\setlength{\tabcolsep}{0pt}
\begin{tabular*}{\textwidth}{@{\extracolsep{\fill}}lll}
\toprule
\emph{Addresses} & \emph{JMLR} & \emph{Moby} \\
\midrule

fellow \epito citizen
&
support \epito vector \epito machin
&
sperm \epito whale
\\
unit \epito state
&
machin \epito learn
&
mobi \epito dick
\\
men \epito women
&
state \epito art
&
mast \epito head
\\
feder \epito govern
&
data \epito set
&
white \epito whale
\\
self \epito govern
&
bayesian \epito network
&
old \epito man
\\
four \epito year
&
larg \epito scale
&
captain \epito ahab
\\
year \epito ago
&
nearest \epito neighbor
&
$G$
\\
american \epito peopl
&
decis \epito tree
&
quarter \epito deck
\\
vice \epito presid
&
cross \epito valid
&
right \epito whale
\\
chief \epito magistr
&
neural \epito network
&
captain \epito peleg
\\

\bottomrule
\end{tabular*}
\end{table}

Our next step is to highlight some episodes that had a high $\rind{G}$
but also ranked low by $\rpart{G}$, and vice versa. In order to do that we sorted episodes based on
\begin{equation}
\label{eq:prop}
	\rho(G) = \frac{\rind{G} - \rpart{G}}{ \rpart{G}}\quad\text{and}\quad
	\eta(G) = \frac{\rpart{G} - \rind{G}}{ \rind{G}}
	\quad.
\end{equation}
The top episodes should have large $\rind{G}$ and $\rpart{G}$ close to $0$.  In
Figure~\ref{fig:bottom5}, we listed top-$5$ episodes from each text dataset.
Many of these episodes contain a true pattern, such as,  \emph{united} $\to$ \emph{states}
or \emph{support} $\to$ \emph{vector} $\to$ \emph{machine} augmented with a common
event, seemingly independent event, such as, \emph{world} or \emph{regression}.
Let us now compare the top-$5$ episodes with large $\eta(G)$. Unlike with $\rho(G)$, this list is dominated
with episodes for which $\supp{G} < \mu_{\text{part}} < \mu_{\text{ind}}$, that is, both methods
overestimate the actual support but the partition model is more correct. To make $\eta(G)$ more  
meaningful, we considered only episodes for which the partition model underestimated the support, $\supp{G} \geq \mu_{\text{part}}$,
given in Figure~\ref{fig:revbottom5}.
We see that the differences between $\rpart{G}$ and $\rind{G}$ are small in Figure~\ref{fig:revbottom5} and large in Figure~\ref{fig:bottom5}.

\begin{figure}[t]
\setlength{\tabcolsep}{0pt}
\emph{Addresses:}\\
\begin{tabular*}{\textwidth}{@{\extracolsep{\fill}}lll}
$G_1$: unit \epito state \epispace world 
&
$G_2$: unit \epito state \epispace shall 
&
\begin{tikzpicture}[baseline]
\node[anchor = base, inner sep = 0pt] (n0) at (-0.3, 0) {$G_3$:};
\node[anchor = base west] (n1) {state};
\node[above = 0.3cm of n1.mid west, anchor = mid west] (n2) {unit};
\node[below right = 0.15cm and 0.5cm of n2.base east, anchor = base west] (n3) {world};
\draw[smalledge, in = 170, out = 0] (n2) edge (n3);
\draw[smalledge, in = 190, out = 0] (n1) edge (n3);
\end{tikzpicture} 
\\
\qquad 9.0  &
\qquad 9.2  &
\qquad 11.4 
\\

\begin{tikzpicture}[baseline]
\node[anchor = base, inner sep = 0pt] (n0) at (-0.3, 0) {$G_4$:};
\node[anchor = base west] (n1) {state};
\node[above = 0.3cm of n1.mid west, anchor = mid west] (n2) {unit};
\node[below right = 0.15cm and 0.5cm of n2.mid east, anchor = mid west] (n3) {peac};
\draw[smalledge, in = 170, out = 0] (n2) edge (n3);
\draw[smalledge, in = 190, out = 0] (n1) edge (n3);
\end{tikzpicture}& 

$G_5$: unit \epito state \epispace peac\\
\qquad 17.8  &
\qquad 17.0 
\\
\end{tabular*}\\[0.5mm]

\emph{JMLR:}\\[0.2mm]
\begin{tabular*}{\textwidth}{@{\extracolsep{\fill}}ll}
$G_1$: support \epito vector \epito machin \epispace regress
&
\begin{tikzpicture}[baseline]
\node[anchor = base, inner sep = 0pt] (n0) at (-0.25, 0) {$G_2$:};
\node[anchor = base west] (n1) {support};
\node[right = 0.5cm of n1.mid east, anchor = mid west] (n2) {vector};
\node[right = 0.5cm of n2.mid east, anchor = mid west] (n3) {machin};
\node[above right = 0.3cm and 0.5cm of n1.mid east, anchor = mid west] (n4) {regress};
\draw[smalledge] (n1.mid east) -- (n2.mid west);
\draw[smalledge] (n2.mid east) -- (n3.mid west);
\draw[smalledge] (n1) edge[in = 180, looseness = 0.5] (n4.mid west);
\end{tikzpicture} \\
\qquad 95.1
&
\qquad 90.4
\\

$G_3$: support \epito vector \epito machin \epispace number
&

\begin{tikzpicture}[baseline]
\node[anchor = base, inner sep = 0pt] (n0) at (-0.25, 0) {$G_4$:};
\node[anchor = base west] (n1) {support};
\node[right = 0.5cm of n1.mid east, anchor = mid west] (n2) {vector};
\node[right = 0.5cm of n2.mid east, anchor = mid west] (n3) {machin};
\node[above right = 0.3cm and 0.5cm of n2.mid east, anchor = mid west] (n4) {regress};
\draw[smalledge] (n1.mid east) -- (n2.mid west);
\draw[smalledge] (n2.mid east) -- (n3.mid west);
\draw[smalledge] (n2) edge[in = 180, looseness = 0.5] (n4.mid west);
\end{tikzpicture} \\
\qquad 52.0
&
\qquad 86.4
\\

$G_5$: support \epito vector \epito machin \epispace space\\
\qquad 51.6
\\
\end{tabular*}\\[0.5mm]

\emph{Moby:}\\[0.2mm]
\begin{tabular*}{\textwidth}{@{\extracolsep{\fill}}lll}
$G_1$: sperm \epito whale \epispace thing
&
$G_2$: sperm \epito whale \epispace ship

\\
\qquad 13.4
&
\qquad 21.3
\\
$G_3$: sperm \epispace whale \epito ship
&\begin{tikzpicture}[baseline]
\node[anchor = base, inner sep = 0pt] (n0) at (-0.25, 0) {$G_4$:};
\node[anchor = base west] (n1) {sperm};
\node[right = 0.6cm of n1.base east, anchor = base west] (n2) {ship};
\node[above right = 0.3cm and 0.6cm of n1.base east, anchor = base west] (n3) {whale};
\draw (n1.mid east) edge[smalledge] (n2.mid west);
\draw (n1) edge[smalledge, in = 180, looseness = 0.5] (n3);
\end{tikzpicture} &
$G_5$: sperm \epispace whale \epispace water \\
\qquad 7.1
&
\qquad 9.5
&
\qquad 14.6
\\

\end{tabular*}
\caption{Episodes with high rank $\rind{G}$ but considered redundant by $\rpart{G}$. Top-5 episodes
based on $\rho(G)$ given in Eq.~\ref{eq:prop}. The given numbers are $\rind{G_i}$, whereas the partition model ranks are $\rpart{G_i} \leq 10^{-6}$.}
\label{fig:bottom5}
\end{figure}

\begin{figure}[t]
\setlength{\tabcolsep}{0pt}
\emph{Addresses:}\\[0.2mm]
\begin{tabular*}{\textwidth}{@{\extracolsep{\fill}}lll}
$G_1$: 
govern \epispace great \epispace world &
$G_2$: 
govern \epispace great \epispace peac &
$G_3$: 
countri \epispace nation \epispace govern
\\
\qquad 1.33 / 1.09 &
\qquad 3.44 / 3.00 &
\qquad 1.16 / 1.02  \\
$G_4$: 
such \epispace nation \epispace peopl &
$G_5$: 
nation \epispace govern \epispace made
\\
\qquad 1.46 / 1.35 &
\qquad 1.62 / 1.52 \\
\end{tabular*}\\[0.5mm]

\emph{JMLR:}\\[0.2mm]
\begin{tabular*}{\textwidth}{@{\extracolsep{\fill}}lll}
$G_1$:
algorithm \epispace show \epispace featur &
$G_2$:
result \epito model \epispace algorithm &

\begin{tikzpicture}[baseline]
\node[anchor = base, inner sep = 0pt] (n0) at (-0.25, 0) {$G_3$:};
\node[anchor = base west] (n1) {show};
\node[above right = 0.3cm and 1.15cm of n1.mid west, anchor = mid west] (n2) {algorithm};
\node[right = 0.5cm of n1.mid east, anchor = mid west] (n3) {problem};
\draw (n1) edge[smalledge, in = 180, looseness = 0.5] (n2);
\draw[smalledge] (n1) edge (n3);
\end{tikzpicture} \\
\qquad 1.07 / 0.36 &
\qquad 1.08 / 0.53 &
\qquad 1.19 / 0.77\\

$G_4$:
algorithm \epispace data \epispace obtain &
$G_5$:
model \epispace train \epispace result &
\\
\qquad 1.08 / 0.75 &
\qquad 1.47 / 1.21\\
\end{tabular*}\\[0.5mm]

\emph{Moby:}\\[0.2mm]
\begin{tabular*}{\textwidth}{@{\extracolsep{\fill}}lll}
$G_1$ :
round \epispace old \epispace whale &
$G_2$ :
now \epispace though \epispace whale &
$G_3$ :
now \epispace round \epispace whale \\

\qquad 4.44 / 3.72 & 
\qquad 1.68 / 1.62 & 
\qquad 1.84 / 1.79 \\

$G_4$ :
out \epispace over \epispace whale &
$G_5$ :
now \epispace whale \epispace good \\

\qquad 1.49 / 1.45 &
\qquad 2.26 / 2.21 \\
\end{tabular*}

\caption{Top-5 episodes based on $\eta(G)$ given in Eq.~\ref{eq:prop}, for which partition model underestimated the support. The number format is $x / y$, where  $x = \rpart{G_i}$ and $y = \rind{G_i}$.}
\label{fig:revbottom5}
\end{figure}

Many of downgraded episodes are \emph{parallel} episodes, for example,
(\emph{united}, \emph{states}), see Table~\ref{tab:top10}.  While the
independence model ranks them high, the elevated support of a \emph{serial}
episode \emph{united} $\to$ \emph{states} explains well the elevated support of
this parallel episode since these words occur almost always in this particular
order. This makes the partition model based on the superepisodes to give this
episode a low score.

\begin{figure}[t]
\begin{tikzpicture}[baseline]
\begin{axis}[xlabel={$\abs{E(\mach{G})}$}, ylabel= {time (ms)},
    height = 2.8cm,
    width = 10cm,
    scale only axis,
    cycle list name=yaf,
	every axis plot post/.style = {},
    xmax = 80, xmin = 1, ymax = 18, ymin = 0,
	legend entries = {upper quartile, median, lower quartile},
	legend pos = south east
    ]

\addplot[only marks, mark size = 1.5, fill opacity = 0.2, mark = o, yafcolor3, mark options = {line width = 0.5pt}]
	table[x index = 0, y index = 5, header = false]  {jmlrtime.dat};
\addplot[only marks, fill opacity = 0.2, mark = +, yafcolor1, mark options = {line width = 0.5pt}]
	table[x index = 0, y index = 3, header = false]  {jmlrtime.dat};
\addplot[only marks, fill opacity = 0.2, mark = x, yafcolor2, mark options = {line width = 0.5pt}]
	table[x index = 0, y index = 4, header = false]  {jmlrtime.dat};

\pgfplotsextra{\yafdrawaxis{1}{80}{0}{18}}
\end{axis}
\end{tikzpicture}
\caption{Time needed to compute the rank for patterns obtained from \emph{JMLR} as a function of the number of edges in $\mach{G}$.}
\label{fig:times}
\end{figure}
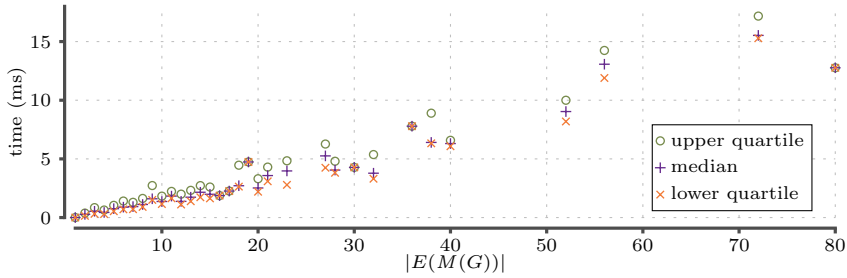

Finally, let us consider running times that are given in Table~\ref{tab:basic}.
We see that we can rank large amount of episodes in a short period of time.
Ranking $50\,000$ episodes took us less than a minute. To obtain a more
detailed picture, we present running times as a function of $\abs{E(\mach{G})}$
in Figure~\ref{fig:times} for \emph{JMLR} episodes. We see that
the more complex episode (the largest episode contained 5 nodes), the longer it takes to rank.
This suggests that
while there are complicated steps in computing the support that may even result
in exponentially large structures, in practice ranking can be done efficiently.

\section{Concluding remarks}\label{sec:conclusions}
In this paper we introduced ranking episodes based on a partition model.  Such
a ranking reduces redundancy among episodes by ranking episodes low if they can
be explained by either two subepisodes or by a more strict episode.

To construct the model we first constructed a finite state machine that is used
for computing the expected support for the independence model. We then modified
the probabilities of some of the transitions. These transitions are selected
based on which subepisodes we are considering. We compare this model to the
independence model and show that for our experiments the model reduces
redundancy in patterns.

The effectiveness of the partition model relies on the assumption that the two
subepisodes (or the superepisode) have few gaps. This causes the parameters
$t_1$ and $t_2$ to be large. If this assumption does not hold, that is, $t_1
\approx t_2 \approx 0$, then the partition model will reduce to the independence
model. While this assumption is natural and reasonable, in a setup where
episodes are frequent but have large gaps, this approach will not reduce
redundancy.  In such a setup, a different approach is needed, a potential
direction for a future line of work.

When partitioning an episode into two subepisodes, we did not consider all the
possible partitions. Instead, we only considered partitions arising from prefix
graphs. While these partitions are a natural subclass of all possible partitions,
this restriction leads to some limitations. For example, we do not
partition a serial episode $a \to b \to c \to d$ to $a \to c$ and $b \to d$.
However, note that for many episodes, every partition is a partition arising
from a prefix graph. This is the case with any parallel episode.  We should
point out that from technical point of view, we can use non-prefix partitions.
However, as demonstrated in Example~\ref{ex:nonprefix} and
Proposition~\ref{prop:nonprefix}, a partition model may not take
properly into account the lack of gaps in a non-prefix subepisode.
Developing a technique
that properly takes interleaving subepisodes into account is an interesting
direction for a future work.

Instead of using just the partition model to rank episodes, it may be
advantageous to combine it with other ranking method. For example, one approach
would be to rank the episodes using the partition model, select top-$k$
episodes, and rerank them based on the independence model. The number $k$ can
be given explicitly or determined by interpreting the rank as a $p$-value, and
filtering the episodes based on a given significance level. In the latter
approach some extra steps are needed, such as adjusting for the multiple
hypotheses testing. This can be done either with direct adjustment or a holdout
approach as described by~\citet{webb:07:significant}.
Strictly speaking, interpreting rank as a $p$-value requires that we know the
exact model parameters which is uncommon.  Consequently, in practice and in
this work we estimate these parameters, and by doing so estimate the true
$p$-value, by finding the maximum likelihood estimates.

This work opens several future lines of research. One straightforward extension
is to combine the partition approach with a markov model suggested
by~\citet{gwadera:05:markov}. A more intriguing extension is to apply this
model for a scenario where we are given one long sequence instead of a database
of sequences. In such a case, the support is either based on sliding windows of
fixed length or minimal windows. Since the instances are no longer independent,
that is, the support is no longer a sum of independent variables, it is likely
that we cannot apply the model directly. However, it may be possible to rank
episodes by using some other statistic than a support.
Table~\ref{tab:top10} for \emph{JMLR} shows that we can still reduce redundancy among the top patterns.
One fruitful approach would be developing a pattern set miner for general
episodes. A potential starting point for such a miner could be SQS
miner~\citep{tatti:12:long}, a pattern set miner for serial episodes.

%
\bibliographystyle{abbrvnat}
\bibliography{abbreviations,bibliography}  

\begin{thebibliography}{17}
\providecommand{\natexlab}[1]{#1}
\providecommand{\url}[1]{\texttt{#1}}
\expandafter\ifx\csname urlstyle\endcsname\relax
  \providecommand{\doi}[1]{doi: #1}\else
  \providecommand{\doi}{doi: \begingroup \urlstyle{rm}\Url}\fi

\bibitem[Achar et~al.(2012)Achar, Laxman, Viswanathan, and
  Sastry]{achar:12:discovering}
A.~Achar, S.~Laxman, R.~Viswanathan, and P.~S. Sastry.
\newblock Discovering injective episodes with general partial orders.
\newblock \emph{Data Mining and Knowledge Discovery}, 25\penalty0 (1):\penalty0
  67--108, 2012.

\bibitem[Gwadera et~al.(2005{\natexlab{a}})Gwadera, Atallah, and
  Szpankowski]{gwadera:05:markov}
R.~Gwadera, M.~J. Atallah, and W.~Szpankowski.
\newblock Markov models for identification of significant episodes.
\newblock In \emph{Proceedings of the 5th SIAM International Conference on Data
  Mining (SDM), Newport Beach, CA}, pages 404--414, 2005{\natexlab{a}}.

\bibitem[Gwadera et~al.(2005{\natexlab{b}})Gwadera, Atallah, and
  Szpankowski]{gwadera:05:reliable}
R.~Gwadera, M.~J. Atallah, and W.~Szpankowski.
\newblock Reliable detection of episodes in event sequences.
\newblock \emph{Knowledge and Information Systems}, 7\penalty0 (4):\penalty0
  415--437, 2005{\natexlab{b}}.

\bibitem[Kullback(1959)]{Kullback59}
S.~Kullback.
\newblock \emph{Information Theory and Statistics}.
\newblock Wiley, 1959.

\bibitem[Lam et~al.(2014)Lam, M{\"{o}}rchen, Fradkin, and
  Calders]{DBLP:journals/sadm/LamMFC14}
H.~T. Lam, F.~M{\"{o}}rchen, D.~Fradkin, and T.~Calders.
\newblock Mining compressing sequential patterns.
\newblock \emph{Statistical Analysis and Data Mining}, 7\penalty0 (1):\penalty0
  34--52, 2014.

\bibitem[Laxman et~al.(2007)Laxman, Sastry, and Unnikrishnan]{laxman:07:fast}
S.~Laxman, P.~S. Sastry, and K.~P. Unnikrishnan.
\newblock A fast algorithm for finding frequent episodes in event streams.
\newblock In \emph{Proceedings of the 13th ACM International Conference on
  Knowledge Discovery and Data Mining (SIGKDD), San Jose, CA}, pages 410--419,
  2007.

\bibitem[Low-Kam et~al.(2013)Low-Kam, Ra\"{\i}ssi, Kaytoue, and
  Pei]{DBLP:conf/icdm/Low-KamRKP13}
C.~Low-Kam, C.~Ra\"{\i}ssi, M.~Kaytoue, and J.~Pei.
\newblock Mining statistically significant sequential patterns.
\newblock In \emph{Proceedings of the 13th IEEE International Conference on
  Data Mining (ICDM), Dallas, TX}, pages 488--497, 2013.

\bibitem[Mannila and Meek(2000)]{mannila:00:global}
H.~Mannila and C.~Meek.
\newblock Global partial orders from sequential data.
\newblock In \emph{Proceedings of the 6th ACM International Conference on
  Knowledge Discovery and Data Mining (SIGKDD), Boston, MA}, pages 161--168,
  2000.

\bibitem[Mannila et~al.(1997)Mannila, Toivonen, and
  Verkamo]{mannila:97:discovery}
H.~Mannila, H.~Toivonen, and A.~I. Verkamo.
\newblock Discovery of frequent episodes in event sequences.
\newblock \emph{Data Mining and Knowledge Discovery}, 1\penalty0 (3):\penalty0
  259--289, 1997.

\bibitem[Pei et~al.(2006)Pei, Wang, Liu, Wang, Wang, and
  Yu]{pei:06:discovering}
J.~Pei, H.~Wang, J.~Liu, K.~Wang, J.~Wang, and P.~S. Yu.
\newblock Discovering frequent closed partial orders from strings.
\newblock \emph{IEEE Transactions on Knowledge and Data Engineering},
  18\penalty0 (11):\penalty0 1467--1481, 2006.

\bibitem[Tatti(2014)]{tatti:14:mining}
N.~Tatti.
\newblock Discovering episodes with compact minimal windows.
\newblock \emph{Data Mining and Knowledge Discovery}, 28\penalty0 (4):\penalty0
  1046--1077, 2014.

\bibitem[Tatti and Cule(2011)]{tatti:11:mining}
N.~Tatti and B.~Cule.
\newblock Mining closed episodes with simultaneous events.
\newblock In \emph{Proceedings of the 17th ACM International Conference on
  Knowledge Discovery and Data Mining (SIGKDD), San Diego, CA}, pages
  1172--1180, 2011.

\bibitem[Tatti and Cule(2012)]{tatti:12:mining}
N.~Tatti and B.~Cule.
\newblock Mining closed strict episodes.
\newblock \emph{Data Mining and Knowledge Discovery}, 25\penalty0 (1):\penalty0
  34--66, 2012.

\bibitem[Tatti and Vreeken(2012)]{tatti:12:long}
N.~Tatti and J.~Vreeken.
\newblock The long and the short of it: summarising event sequences with serial
  episodes.
\newblock In \emph{Proceedings of the 18th ACM International Conference on
  Knowledge Discovery and Data Mining (SIGKDD), Beijing, China}, pages
  462--470, 2012.

\bibitem[Wang and Han(2004)]{wang:04:bide}
J.~Wang and J.~Han.
\newblock Bide: Efficient mining of frequent closed sequences.
\newblock In \emph{Proceedings of the 20th International Conference on Data
  Engineering (ICDE), Boston, MA}, pages 79--90, 2004.

\bibitem[Webb(2007)]{webb:07:significant}
G.~I. Webb.
\newblock Discovering significant patterns.
\newblock \emph{Machine Learning}, 68\penalty0 (1):\penalty0 1--33, 2007.

\bibitem[Webb(2010)]{webb:10:self-sufficient}
G.~I. Webb.
\newblock Self-sufficient itemsets: An approach to screening potentially
  interesting associations between items.
\newblock \emph{ACM Transactions on Knowledge Discovery from Data}, 4\penalty0
  (1), 2010.

\end{thebibliography}
%
\appendix

\section{Proof of Proposition~\lowercase{\ref{prop:greedy}}}
In order to prove the proposition we need the following lemma
which we will state without the proof.

\begin{lemma}
Assume that a
sequence $S = s_1, \ldots, s_n$ covers an episode $G$. If there is a source
vertex $v$ such that $s_1 = \lab{v}$, then $s_2, \ldots, s_n$ covers $G
\setminus v$.  Otherwise, $s_2, \ldots, s_n$ covers $G$. 
\end{lemma}

\begin{proof}[of Proposition~\ref{prop:greedy}]
We need to prove only ''only if'' case. Assume that $S = s_1, \ldots, s_n$
covers an episode $G$.

We will prove the proposition by induction over $n$.  Obviously, the result
holds for $n = 0$. Write $S' = s_2, \ldots, s_n$.

If there is no source vertex in $G$ with a label $s_1$, then $\greedy{M, S} =
\greedy{M, S'}$. Now the lemma implies that $S'$ covers $G$ and the induction
assumption implies that $\greedy{M, S'} = G$.

If there is a source vertex $v$ in $G$ such that $\lab{v} = s_1$, then
$\greedy{M, S} = \greedy{M, S', G(v)}$. Note that the $G(v)$ and its
descendants form exactly $\mach{H}$, where $H = G \setminus v$. That is,
$\greedy{M, S} = G$ if and only if $\greedy{M(H), S'} = H$. The lemma implies
that $S'$ covers $H$ and the induction assumption implies that $\greedy{M(H), S'} = H$
which proves the proposition.
\qed\end{proof}

\section{Proof of Proposition~\lowercase{\ref{prop:concave}}}

In order to prove the proposition we need the following proposition, which
essentially describes the properties of a log-likelihood of a log-linear model.
The proof of this proposition can be found, for example, in~\cite{Kullback59}.

\begin{proposition}
\label{prop:loglinear}
Assume that we are given a set of $k$ functions $\funcdef{T_i}{\Omega}{\reals}$, mapping
an object from some space $\Omega$ to a real number.
For $n$ real numbers, $r_1, \ldots, r_k$, define
\[
	Z(r_1, \ldots, r_k) = \sum_{\omega \in \Omega} \exp{\sum_{i = 1}^k r_iT_i(\omega)}\quad.
\]
Define a distribution
\[
	p(\omega) = \frac{\exp{\sum_{i = 1}^k r_iT_i(\omega)}}{Z(r_1, \ldots, r_k)}\quad.
\]
Let $X$ be a multiset of events from $\Omega$.  Define
\[
	c(r_1, \ldots, r_k) = \sum_{\omega \in X} \log p(\omega)\quad.
\]
Then $c$ is a concave function of $r_1, \ldots, r_k$. In fact
\[
	\frac{\partial c}{\partial  r_i} = \sum_{\omega \in X} (T_i(\omega) - \mean{p}{T_i})
\]
and
\[
	\frac{\partial c}{\partial  r_i r_j} = \abs{X}(\mean{p}{T_i}\mean{p}{T_j} - \mean{p}{T_i T_j})\quad.
\]
\end{proposition}

\begin{proof}[of Proposition~\ref{prop:concave}]
In order to prove the result we need to rearrange the terms
in $\log p(\mathcal{S})$ based on current state. In order to do that,
let us define $L_H$ to be a multiset of labels that occur in $\mathcal{S}$ while
the current state is $H$, that is,
\[
	L_H = \bigcup_{s_1, \ldots, s_n = S_i \atop i = 1, \ldots, m} \set{s_j \mid \greedy{s_1, \ldots, s_{j - 1}} = H}\quad.
\]

We can now rewrite the log-likelihood as
\begin{equation}
\label{eq:terms}
	\log p(\mathcal{S}) = \sum_{H \in V(M)} \sum_{l \in L_H} \log p(l \mid H)\quad.
\end{equation}

All we need to show now is that each term can be expressed in the form given in Proposition~\ref{prop:loglinear}.
In order to do that, define for each label $l$ an indicator function
\[
	T_l(s) =
	\begin{cases}
		1, & \text{if } l = s, \\
		0, & \text{otherwise}\quad. 
	\end{cases}
\]
Also, define indicator functions whether the transition is in $C_1$ or $C_2$,
that is, define $T_1$ and $T_2$ as
\[
	T_i(s) =
	\begin{cases}
		1, & \text{if there is } (H, F) \in C_i \text{ with } \lab{H, F} = s , \\
		0, & \text{otherwise}\quad. 
	\end{cases}
\]
We have now
\[
	p(l \mid H) = \frac{1}{Z_H}\exp\bigg(t_1T_1(l) + t_2T_2(l) + \sum_{s \in \Sigma} u_sT_s(l)\bigg)\quad.
\]
Since $Z_H$ corresponds exactly to the normalization constant in Proposition~\ref{prop:loglinear},
we have shown that
\[
	\sum_{l \in L_H} \log p(l \mid H)
\]
is a concave function. The sum of concave functions is concave, proving the result. 
\qed\end{proof}

The proof also reveals how to compute the gradient and the Hessian matrix.
These are needed if we are optimize $\log p(\mathcal{S})$.  Since $\log
p(\mathcal{S})$ is a sum of functions given in Proposition~\ref{prop:loglinear}
the gradient and the Hessian matrix of $\log p(\mathcal{S})$ can be obtained
by summing gradients and Hessian matrices of individual terms of Equation~\ref{eq:terms}.

\section{Proof of Proposition~\lowercase{\ref{prop:partcascade}}}
\begin{proof}
Let $F = \greedy{M, s_1, \ldots, s_{n - 1}}$.

If $F = H$,
then we remain in $H$ only if $s_n$ is not a label of an outgoing edge.
The probability of this is equal to $q$.

If $F \neq H$, the only way $\greedy{M, S} = H$, is that $F$ is a parent of $H$
and the label connecting $F$ to $H$ is equal to $s_n$. This gives us the result.
\qed\end{proof}

\section{Computing gradient descent}\label{sec:descent}

We use Newton-Raphson method to fit the model.
In order to do this we need to compute the gradient and the Hessian matrix with respect to the parameters.
This can be done efficiently as described by the following proposition.

\begin{proposition}
\label{prop:gradient}
Let $G$ be an episode and let $M = \mach{G}$. Let $H$ be a state in $M$.

Let $C_1$ and $C_2$ be two disjoint subsets of $E(M)$.
Define $L_i$ to be the set of labels such that $l \in L_i$ if and only if there
is an edge $(H, F) \in C_i$ labelled as $l$.
Let $J$ be a matrix of size $2 \times \abs{\Sigma}$ such that
$J_{il} = 1$ if $l \in L_i$, and $0$ otherwise.

Let $v$ be a vector of length $\abs{\Sigma}$ such that $v_l = p(l \mid H)$
is equal to the probability of generating label $l$. Define $w = Jv$.

Let $c$ be the count of how often we stay in $H$, 
\[
	c = \abs{\set{(i, j) \mid s = S_j, H = \greedy{M, (s_1, \ldots, s_i)}}}\quad.
\]

Let $n$ be a vector of length $\abs{\Sigma}$, 
\[
	n_l = \abs{\set{(i, j) \mid s = S_j, H = \greedy{M, (s_1, \ldots, s_{i - 1})}, s_i = l}},
\]
to contain the number of symbols labelled as $l$ visited in $\mathcal{S}$ while being in the state $H$.

Let $V = \diag{v}$ and let $W = \diag{w}$.  Define
\[
d_H = 
\spr{
\begin{matrix}
n - cv \\
J(n - cv) \\
\end{matrix}}
\quad\text{and}\quad
B_H = 
c\spr{
\begin{matrix}
V - vv^T& VJ^T - vw^T \\
JV - wv^T & W - ww^T\\
\end{matrix}} \quad.
\]
Then the gradient and hessian of $\log p(\mathcal{S})$ at $\set{u_i}$, $t_1$ and $t_2$ is equal to 
\[
	d = \sum_{H \in V(M)} d_H \quad\text{and}\quad B = \sum_{H \in V(M)} B_H\quad.
\]
\end{proposition}

\begin{proof}
Proposition~\ref{prop:loglinear} and Proposition~\ref{prop:concave} imply
that the gradient of $q_H = \sum_{l \in L_H} \log p(l \mid H)$ is equal to $d_H$ and
the hessian is equal to $B_H$. Since $\log p(\mathcal{S}) = \sum_{H \in V(M)} q_H$, the result follows.
\qed\end{proof}

In order to obtain additional speed-ups, first notice that
Proposition~\ref{prop:gradient} implies that we do need to scan the original
sequence set every time. Instead it is enough to compute the vector $n$ and a scalar $c$
for each state $H$. Moreover, for a fixed episode $G$, the rank does not depend on
probabilities of individual labels that do not occur in $G$. In other words, we can
treat all labels that do not occur in $G$ as one label. This will reduce the length
of the gradient and the size of the hessian from $\abs{\Sigma} + 2$ to $\abs{V(G)} + 3$, at most.
These speed-ups make solving the model very fast in practice.

\end{document}